\RequirePackage{fix-cm}
\documentclass[smallextended]{svjour3}       % onecolumn (second format)
\smartqed  % flush right qed marks, e.g. at end of proof
\usepackage{amssymb,latexsym,amsfonts,amsmath,mathrsfs}
\usepackage{graphicx}
\usepackage{geometry}
\usepackage{dsfont}
\usepackage{setspace}
\usepackage{tabularx}
\usepackage{placeins}
\usepackage{hyperref}
\usepackage{subfigure}

\hypersetup{colorlinks=true, linkcolor=blue, citecolor=red}

\usepackage{appendix}
\usepackage{ctable}

\allowdisplaybreaks

\newcommand{\xqedhere}[2]{%
  \rlap{\hbox to#1{\hfil\llap{\ensuremath{#2}}}}}

\begin{document}

\title{Dynamics of Virus and Immune Response in Multi-Epitope Network}

%\titlerunning{Short form of title}        % if too long for running head

\author{Cameron J. Browne     \and     Hal L. Smith
}

%\authorrunning{Short form of author list} % if too long for running head

\institute{C.J. Browne \at Mathematics Department, University of Louisiana at Lafayette, Lafayette, LA \\
            \email{cambrowne@louisiana.edu}    \and  H.L. Smith \at School of Mathematical \& Statistical Sciences, Arizona State University      %  \\
%             \emph{Present address:} of F. Author  %  if needed
  }

\maketitle

\begin{abstract}
The host immune response can often efficiently suppress a virus infection, which may lead to selection for immune-resistant viral variants within the host.  For example, during HIV infection, an array of CTL immune response populations recognize specific epitopes (viral proteins) presented on the surface of infected cells to effectively mediate their killing. However HIV can rapidly evolve resistance to CTL attack at different epitopes, inducing a dynamic network of interacting viral and immune response variants.  We consider models for the network of virus and immune response populations, consisting of Lotka-Volterra-like systems of ordinary differential equations.  Stability of feasible equilibria and corresponding uniform persistence of distinct variants are characterized via a Lyapunov function.  We specialize the model to a ``binary sequence'' setting, where for $n$ epitopes there can be $2^n$ distinct viral variants mapped on a hypercube graph.  The dynamics in several cases are analyzed and sharp polychotomies are derived characterizing persistent variants.  In particular, we prove that if the viral fitness costs for gaining resistance to each epitope are equal, then the system of $2^n$ virus strains converges to a ``perfectly nested network'' with less than or equal to $n+1$ persistent virus strains.  Overall, our results suggest that \emph{immunodominance}, i.e. relative strength of immune response to an epitope, is the most important factor determining the persistent network structure.
% \PACS{PACS code1 \and PACS code2 \and more}
% \subclass{MSC code1 \and MSC code2 \and more}
\end{abstract}

\section{Introduction}\label{s2}

The dynamics of virus and immune response within a host can be viewed as a complex ecological system.  Both predator-prey and competitive interactions are especially important during a host infection.  The immune response predates on the pathogen, and distinct viral strains compete for a target cell population, while immune response populations compete for the virus since their proliferation occurs upon pathogen recognition.   The immune response can cause significant mortality of the virus, which may lead to selection for immune-resistant viral variants within the host.  For example during HIV infection, an extensive repertoire of CTL immune effectors recognize specific epitopes (viral proteins) presented on the surface of infected cells to effectively mediate their killing, however HIV can rapidly evolve resistance to CTL attack at different epitopes.  The ensuing battle precipitates a dynamic network of interacting viral strains and immune response variants, analogous to an ecosystem of rapidly evolving prey countering attack from a diverse collection of predators.

While the virus-immune interactions may be quite complex, patterns and structure can emerge.  For the cellular immune response, a consistent and reproducible hierarchy of T cell populations organize in response to multiple epitopes of a pathogen, according to their (vertical) \emph{immunodominance}, i.e. relative expansion levels of the responding immune populations within the host \cite{Kloverpris}.  Vertical T cell
immunodominance patterns are highly variable among HIV infected individuals
and change over time, largely due to sequence
variability in the viral ``quasispecies'' \cite{liu2013vertical}.  Rapidly evolving pathogens, such as HIV and HCV, can evade the immune response via mutations at multiple epitopes.  The pattern of epitope mutations, called the \emph{escape pathway}, is of significant interest, and there is some evidence that the viral evolution is predictable \cite{Barton}.  The fitness of an emerging viral mutant strain, along with the strength of the CTL response, certainly affect the selection pressure for a single epitope mutation \cite{Vitaly3}.  However the concurrent interaction of diverse virus and immune response populations necessitate considering the whole system together in order to understand viral escape of multiple epitopes \cite{liu2013vertical}.  In this paper, we introduce and analyze mathematical models for the dynamics of virus and immune response in a network determined by interaction at multiple epitopes.

A large amount of work on modeling within-host virus dynamics has been based on the ``standard'' virus model; an ordinary differential equation system describing the coupled changes in target cells, infected cells, and free virus particles through time in an infected individual \cite{Perelson2}.  The CTL immune response has been included in variations of the standard virus model by considering an immune effector population which kills and is activated by infected cells according to a mass-action (bilinear) rate, although other functional forms for the activation rate have been utilized \cite{browne2016,deBoer}.  Nowak et al. \cite{nowak1996population}, along with other subsequent works \cite{bobko2015singularly,iwasa2004some,souza2011global}, have considered the dynamics of multiple virus strains which are attacked by strain-specific CTL immune response populations (``one-to-one'' virus-immune network).   However the assumption of strain-specific immune response does not correspond to the biological reality that CTLs are specific to epitopes and, in general, multiple epitopes will be shared among virus strains.

Multi-epitope models have been utilized with different datasets cataloging several epitope specific CTL response and viral escape mutations, in order to quantify escape rates and patterns \cite{liu2013vertical,Leviyang,pandit2014reliable}.  Earlier work often considered escape dynamics at epitopes separately, but some recent work has emphasized the concurrent interaction of distinct CTLs with the virus at multiple epitopes.    Ganusov et al. \cite{Vitaly1,Vitaly2} explicitly include multiple CTL clones specific to different epitopes in the standard virus model, and utilize statistical approaches on a linearized version of the model to estimate rates of escape.  Althaus et al. and van Deutekom et al. have also considered multiple epitopes and viral strains in the standard virus model, although results were mostly based on stochastic simulations \cite{Althaus,vanDeutekom}.

Browne \cite{browne2016global} recently analyzed the stability and uniform persistence of a multi-epitope virus-immune model with a \emph{perfectly nested} interaction network. The model setup mirrors a tri-trophic chemostat  ecosystem with a single resource (healthy cells), and a network of consumers (viral strains) and their predators (immune variants).  The perfectly nested network constrains the viral escape pathway so that resistance to multiple epitopes is built sequentially in the order of the immunodominance hierarchy.
  The successive rise of more broadly resistant prey (coming with a fitness cost) and weaker but more generalist predators, in a perfectly nested fashion, is the route to persistence of nested bacteria-phage  communities argued in \cite{jover2013mechanisms,korytowski2015nested}.  The analysis of more complex interaction networks, which allow arbitrary viral escape pathways for building multi-epitope resistance, will be addressed in this paper.

Here we extend the previous work by analyzing a within-host virus model with a general interaction network of multiple variants of virus and immune response. Our results suggest that a diverse viral quasispecies is constructed by resistance mutations at multiple epitopes and the immunodominance hierarchy is the main factor shaping the viral escape pathway.  These notions are supported by observations in HIV infection.  Indeed, the efficacy and breadth of cognate CTL immune responses increase within-host HIV diversity, driving viral evolution so that different combinations of multiple epitope escapes become prevalent in the viral population \cite{pandit2014reliable}.  Also, recent studies have shown that immunodominance hierarchies in HIV are major determinants of viral escape from multiple epitopes \cite{Barton,liu2013vertical}, in particular immunonodominance was found to play a substantially larger role than the viral fitness costs and other factors \cite{liu2013vertical}.
 Understanding the main factors shaping viral escape pathways and immune dynamics is important for design of effective vaccines and immunotherapies.
 
  The outline of this paper is as follows.  In Section \ref{SecModel}, we formulate our general virus-immune network model.  We also introduce the ``binary sequence'' example motivated by HIV and CTL immune response dynamics, where a viral strain is either completely susceptible $(0)$ or has evolved complete resistance $(1)$ to immune attack at a specific epitope, in which for $n$ epitopes, there can be $2^n$ distinct viral variants distinguished by their immune resistance profile.  In Section \ref{Sec3}, we characterize the structure of feasible equilibria in the general model, along with finding a Lyapunov function for stability and corresponding uniform persistence of distinct variants.  Next, in Section \ref{Sec5} we analyze the \emph{binary sequence} example, deriving some graph-theoretic properties of feasible equilibria, and classifying dynamics in several special cases.  In particular, if we constrain the virus-immune response network to be  ``perfectly nested'' (Sec. \ref{NestedSec}), ``strain-specific'' (Sec. \ref{ssSec}), or have $n=2$ epitopes (Sec. \ref{TwoEpitope}), sharp polychotomies characterize persistent variants.  In Section \ref{EqFitSec} we prove that if the viral fitness costs for gaining resistance to each epitope are equal, then the system of $2^n$ virus strains converges to a perfectly nested network with less than or equal to $n+1$ persistent virus strains.

\section{Mathematical model} \label{SecModel}

We consider the following general virus-immune dynamics model, as in Browne \cite{browne2016global}, which includes a population of target cells ($X$), $m$ competing virus strains ($Y_i$ denotes strain $i$ infected cells), and $n$ variants of immune response ($Z_j$):
\begin{align}
\frac{dX}{dt} &= b-cX- X\sum_{i=1}^m \beta_iY_i, \notag \\
\frac{dY_i}{dt} &= \beta_iY_iX-\delta_iY_i-Y_i\sum_{j=1}^m r_{ij} Z_j , \quad i=1,\dots,m \label{ode2}  \\
  \frac{dZ_j}{dt} &= q_jZ_j\sum_{i=1}^n r_{ij}Y_i-  \mu_jZ_j, \notag \quad j=1,\dots,n.
\end{align}

 The function $f(X)=b-cX$ represents the net growth rate of the uninfected cell population.  The parameter $\beta_i$ is the infection rate and $\delta_i$ is the decay rate for infected cells infected with virus strain $i$.  The parameter $\mu_j$ denotes the decay rate of the immune response population $j$.  We assume immune killing and activation rates are mass-action, representative of these events occurring as immune response cells recognize epitopes on the surface of infected cells.  The parameter $r_{ij}$ describes the killing/interaction rate of immune population $Z_j$ on a strain-$i$ infected cell, whereas $q_jr_{ij}$ describes the corresponding activation rate for $Z_j$ (proportional to interaction rate $r_{ij}$).  In the present paper, we assume that virus load (the abundance of virions) is proportional to the amount of (productively) infected cells.  This assumption has frequently been made for HIV since the dynamic of free virions occurs on a much faster time scale than the other variables.  

 The model can be rescaled by introducing the following quantities:
\begin{align*}
x&=\frac{c}{b}X, \quad y_i=\frac{\delta_i}{b}Y_i, \quad \tau=ct, \\
a_{ij}&=\frac{r_{ij}}{\delta_i}, \quad\gamma_i=\frac{\delta_i}{c}, \quad \sigma_j=\frac{\mu_j}{ c}, \\
\mathcal R_i&=\frac{b\beta_i}{c\delta_i},  \quad \rho_j=\frac{\mu_j}{b q_j},
\end{align*}
 The model then becomes:
\begin{align}
\dot x &= 1-x- x\sum_{i=1}^m \mathcal R_i y_i, \notag \\
\dot y_i &= \gamma_i y_i\left(\mathcal R_i x -1 -\sum_{j=1}^n a_{ij}Z_j \right), \quad i=1,\dots,m \label{ode3}  \\
  \dot Z_j &= \frac{\sigma_j}{\rho_j} Z_j\left(\sum_{i=1}^m a_{ij} y_i - \rho_j \right), \quad j=1,\dots,n. \notag
\end{align}

Here $\mathcal R_i$ represents the basic reproduction number of virus strain $i$.  Note that $\rho_j$ represents the reciprocal of the immune response fitness excluding the (rescaled) avidity to each strain $j$.   The $m\times n$ nonnegative matrix $A=\left(a_{ij}\right)$ describes the virus-immune \emph{interaction network}, which determines each immune effector population's avidity to the distinct viral strains.   Each virus strain $i$ (cells infected with strain $i$), $y_i$, has a set of CTLs, $z_j$, that recognize and attack $y_i$.  We call this set the \emph{epitope set of} $y_i$, denoted by $\Lambda_i$, where \begin{align}
\Lambda_i:=\left\{ j\in [1,n] : a_{ij}>0 \right\}, \label{EpSet}
\end{align}
i.e. $j\in \Lambda_i$, if $y_i$ is \emph{not} completely resistant to CTL $Z_j$.

It is not hard to show that solutions to \eqref{ode3} remain non-negative and bounded for all time $t$ \cite{browne2016global}. In Section \ref{Sec3}, we analyze feasible equilibria and their stability in the model \eqref{ode3}.  Note that the general interaction network in \eqref{ode3} allows for \emph{cross-reactivity}, i.e. targeting of multiple epitopes by an immune population $Z_j$.  However in Section \ref{Sec5}, we focus on a special case of the interaction network for epitope specific CTL immune responses, which we introduce below.

\subsection{Example: ``Binary Sequence Case''}\label{BMcase}
Here we specialize system \eqref{ode3} in order to model viral escape from multiple epitopes targeted by specific CTL immune responses, as occurs during HIV infection \cite{Barton,liu2013vertical}.  Suppose there are $n$ viral epitopes, each one recognized by their corresponding specific CTL variant population from the set $Z_1,\dots Z_n$.  To model the viral \emph{escape pathway}, we consider two possible alleles for each epitope: the wild type (0) and the mutated type (1) which has escaped recognition from the cognate immune response.   For each virus strain $y_i$, we associate a binary sequence of length $n$,  $y_i\sim\mathbf i=\left(i_1,i_2,\dots,i_n\right)\in\left\{0,1\right\}^n$, coding the allele type at each epitope.   We assume that each immune response ($Z_j$) targets its specific epitope at the specific rate $a_j$ for virus strains containing the wild-type (allele 0) epitope $j$, whereas $Z_j$ completely loses ability to recognize strains with the mutant (allele 1) epitope $j$, i.e.
\begin{align}
\forall i\in[1,m]: \quad a_{ij}=a_j>0 \ \  \text{if} \ \  j\in\Lambda_i, \label{assumeaj}
\end{align} 
 where $\Lambda_i$ is the strain $i$ (susceptible) \emph{epitope set} defined earlier for model \eqref{ode3} by \eqref{EpSet}  (see Fig. \ref{fig2}).  For example, the wild-type (founder) virus strain, denoted here by $y_w$, is represented by the sequence of all zeroes and epitope set $\Lambda_w=\left\{1,\dots,n\right\}$ since it is susceptible to attack by all immune responses. With assumption \eqref{assumeaj}, we can define an \emph{immune reproduction number} corresponding to each $Z_j$: 
 \begin{align}
 \mathcal I_j:=\frac{a_j}{\rho_j} .\label{IR0}
 \end{align} 
 Furthermore there are $m=2^n$ possible viral mutant strains distinguished by reproduction number $\mathcal R_i$ and epitope set $\Lambda_i$ (or equivalently the binary string $\mathbf i\in\left\{0,1\right\}^n$).   System \eqref{ode3} can be rescaled as:
\begin{align}
\dot x &= 1-x- x\sum_{i=1}^m \mathcal R_i y_i, \quad \dot y_i = \gamma_i y_i\left(\mathcal R_i x -1 -\sum_{j\in \Lambda_i}  z_j \right), \quad  \dot z_j = \frac{\sigma_j}{s_j} z_j\left(\sum_{i: j\in \Lambda_i} y_i -s_j \right), \label{odeS}
\end{align}
where $z_j=a_j Z_j$, $s_j=1/\mathcal I_j$, $i=1,\dots,m\leq 2^n$ and $j=1,\dots,n$.

 \begin{figure}[t]
\subfigure[][]{\label{fig2a}\includegraphics[width=9cm,height=3.5cm]{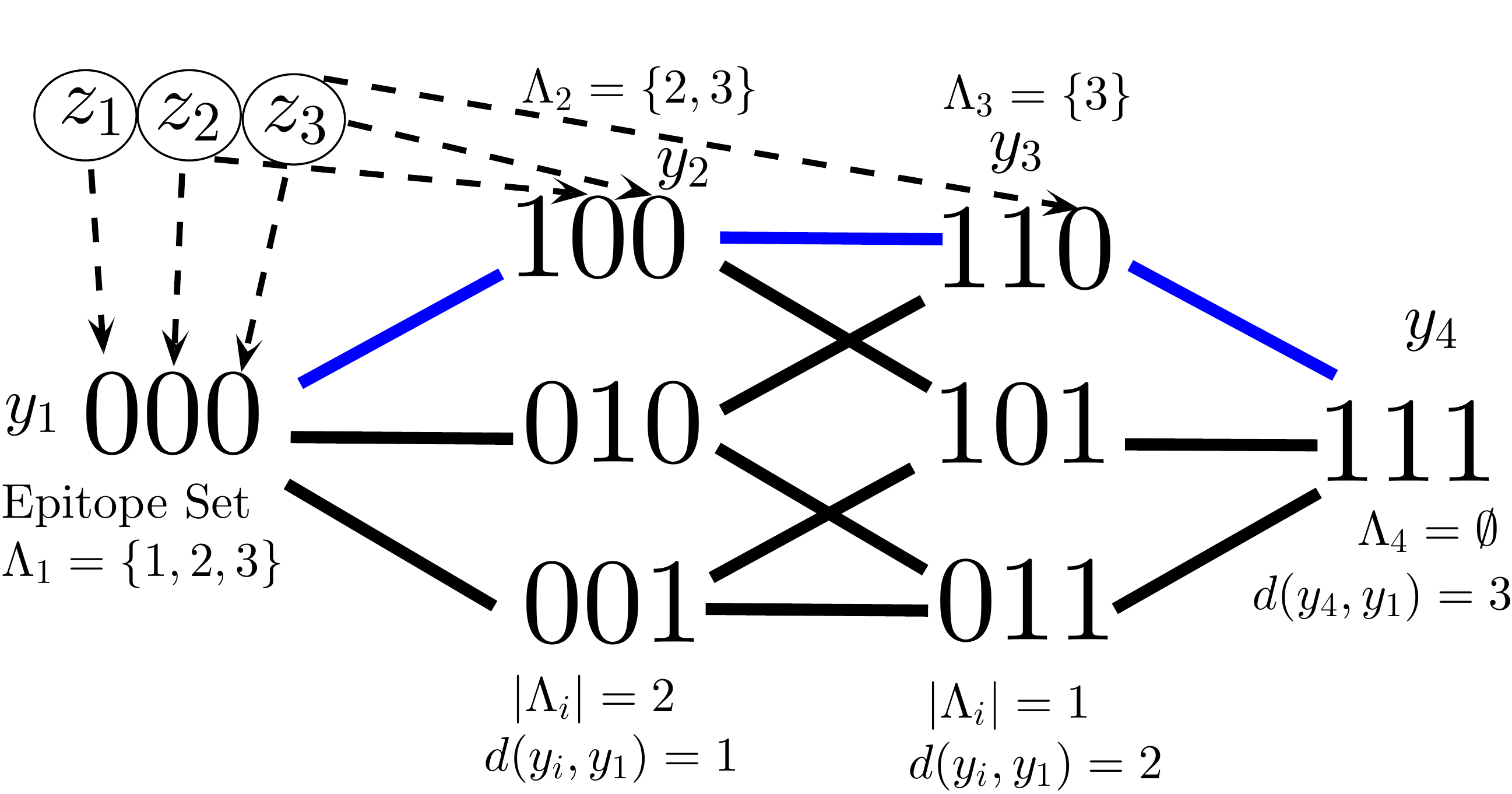}} \qquad
\subfigure[][]{\label{fig2b}\includegraphics[width=3.2cm,height=3.2cm]{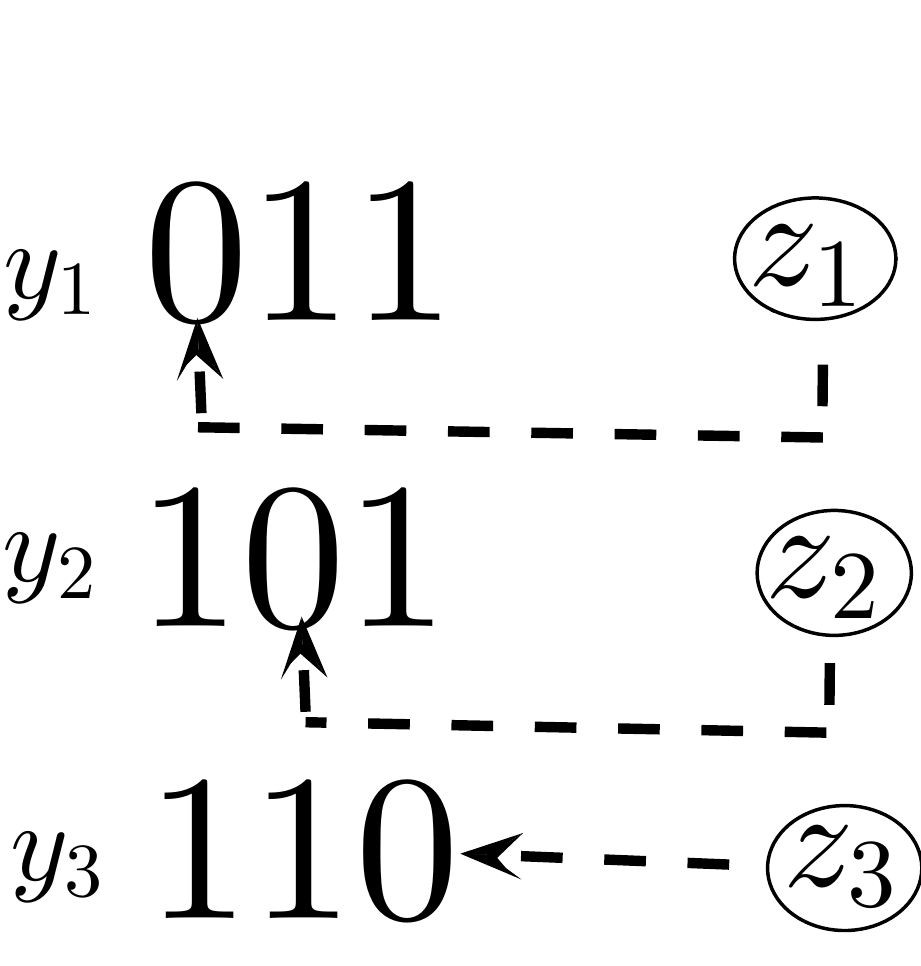}} \\
\subfigure[][]{\label{fig2c}\includegraphics[width=7.6cm,height=2.5cm]{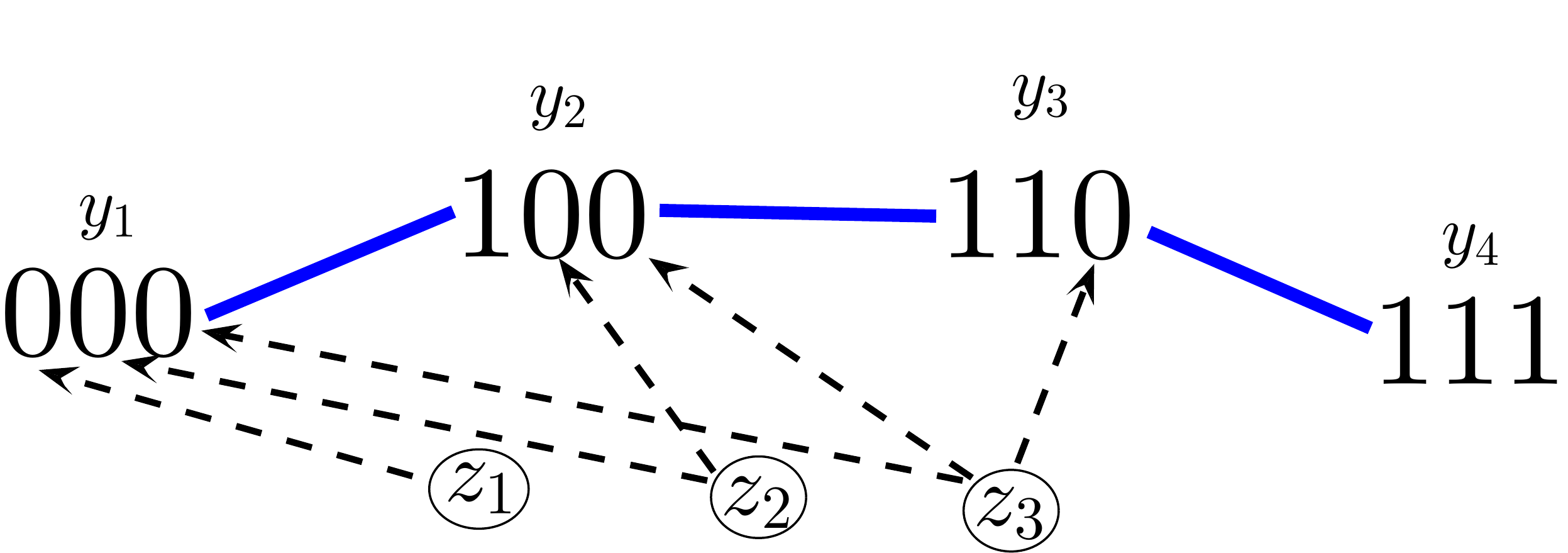}}
\qquad
\subfigure[][]{\label{fig2d}\includegraphics[width=5.6cm,height=2.5cm]{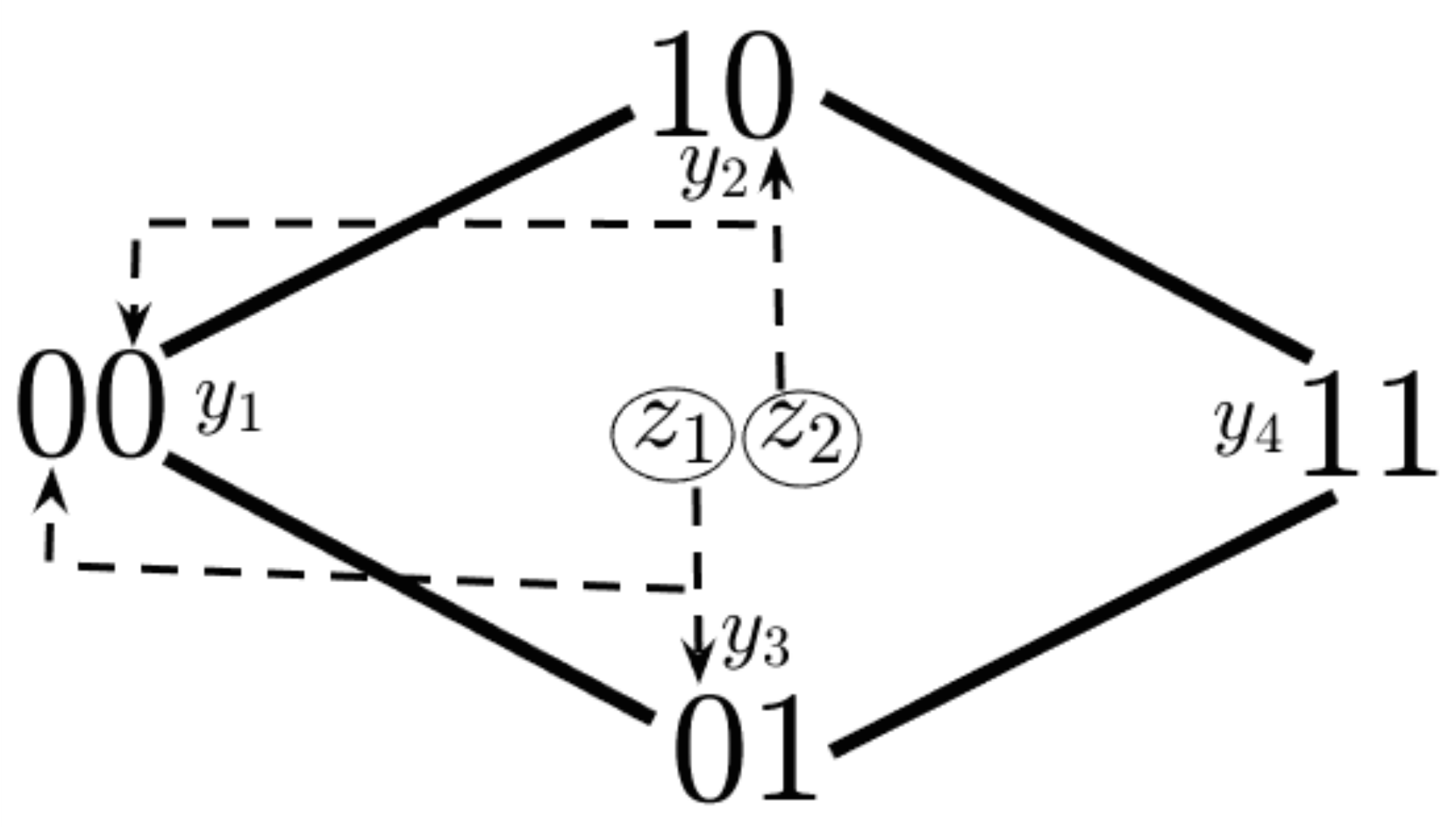}}
\caption{\textbf{(a)}  The full virus-immune network on $n=3$ epitopes for model \eqref{odeS} visualized through the viral immune escape pathway in the hypercube graph, $Q_3$.  Here each viral strain $y_i$, $i=1,\dots,8$, is associated with a unique binary string $\mathbf i\in\left\{0,1\right\}^3$ coding their allele type, susceptible (0) or resistant (1), at each epitope.  Immune response $z_j$ attacks $y_i\sim \mathbf i$ if $i_j=0$, or equivalently if $j$ is in epitope set of $y_i$ ($j\in\Lambda_i$). The wild-type virus, $y_1\sim 000$, can evolve resistance to each epitope-specific immune response $z_j$ by successive single epitope mutations forming a path in the hypercube graph to the completely resistant viral strain ($111$).  The number of epitope mutations which viral strain $y_i$ has accumulated, $d(y_i,y_1)$, is the Hamming distance between $\mathbf i$ and $000$.  The particular escape path with successive epitope mutations in order of immunodominance forms the ``perfectly nested'' network (highlighted in blue).  Note that system \eqref{odeS} does not explicitly include mutation between viral strains, and for figure clarity we only display interaction arrows between immune response and viral variants in the nested subnetwork.   \textbf{(b)}  The strain-specific (one-to-one) network, as a subgraph of the hypercube graph.  \textbf{(c)}  The perfectly nested network, as a subgraph of the hypercube graph. \textbf{(d)}  The full network on $n=2$ epitopes. }
  \label{fig2}
  \end{figure}

Although mutation rates are not explicitly included in model \eqref{odeS}, the $2^n$ potential virus strains can be viewed in a mutational pathway network; each strain is a vertex in an $n$-dimensional \emph{hypercube graph}, shown in Figs. \ref{fig2a} and \ref{fig2d} in the case of $n=3$ and $n=2$ epitopes.  Viral strains $y_i\sim \mathbf i= (i_1,\dots, i_n)$ and $y_k\sim \mathbf k=(k_1,\dots k_n)$ are connected by an edge, which we denote $y_i\leftrightarrow y_k$,  if the sequences $\mathbf i$ and $\mathbf k$ differ in exactly one bit, i.e. their \emph{Hamming distance} -- denoted by $d(y_i,y_k)$ -- is one.   In this way, $y_i$ can mutate into $y_k$ (or vice-versa) through a single epitope mutation if $y_i\leftrightarrow y_k$.   Since each mutation of an epitope comes with a fitness cost, we assume that 
\begin{align}\label{fitnesscost}
\text{If} \ y_i\leftrightarrow y_k \ \text{and} \ d(y_i,y_w)<d(y_k,y_w), \ \text{then} \  \mathcal R_i > \mathcal R_k .
\end{align}
 We say that an immune response $z_j$ is \emph{immunodominant} over another immune response $z_k$ if $\mathcal I_j>\mathcal I_k$ and assume without loss of generality the ordered \emph{immunodominance hierarchy}; $\mathcal I_1\geq \mathcal I_2\geq \dots \geq \mathcal I_n$.

System \eqref{odeS} generalizes many previous model structures in the sense that they can be seen as subgraphs of our ``hypercube network''.  For instance, the ``strain-specific'' (virus-immune response) network \cite{nowak1996population} (also called ``one-to-one network'' in phage-bacteria models \cite{jover2013mechanisms,korytowski2015nested}) is equivalent to restricting \eqref{odeS}  to the $m=n$ viral strains which have mutated $n-1$ epitopes (Figure \ref{fig2b}).  The ``perfectly nested network'' restricts \eqref{odeS} to the $m=n+1$ viral strains which have sequential epitope mutations in the order of the immunodominance hierarchy (Figure \ref{fig2b}).  Nested networks were considered in HIV models \cite{browne2016global}, along with phage-bacteria models \cite{korytowski2015nested}, and may be a common persistent structure in ecological communities \cite{gurney2017network}.  The ``full hypercube network'' has been considered for modeling CTL escape patterns in HIV infected individuals \cite{Althaus,vanDeutekom}.    The dynamics of system \eqref{odeS} with the full network on $n$ epitopes (consisting of $m=2^n$ viral strains), along with one-to-one and perfectly nested subgraphs (consisting of $m=n$ or $m=n+1$ strains), will be  analyzed in Section ~\ref{Sec5}.

\section{Equilibria and Lyapunov function}\label{Sec3}
A general non-negative equilibrium point, $\mathcal E^*=\left(x^*,y^*, Z^* \right)\in\mathbb R_+^{1+m+n}$, of system (\ref{ode3}) will be characterized in terms of the positive virus and immune variant components.  Define the ``\emph{persistent variant sets}'' associated with $\mathcal E^*$ as:
\begin{align}
\Omega_y= \left\{i \in [1,m]: y^*_i>0 \right\} \quad \text{and} \quad \Omega_z= \left\{j \in [1,n]: Z^*_j>0 \right\}.
\end{align}
  In addition, define the following subsets of $\mathbb R_+^{1+m+n}$:
 \begin{align}
 \Omega &= \left\{ \left(x,y, Z\right)\in \mathbb R_+^{1+m+n} \ | \  y_i,z_j>0 \ \text{if} \ i \in\Omega_y, j\in\Omega_z\right\}, \  \Gamma_{\Omega} = \Omega \cap \left\{    y_i,z_j=0, \ \ i\notin\Omega_y,j\notin\Omega_z\right\}.
 \end{align}
 Here $\Gamma_{\Omega}$, consisting of only those state vectors having the same set of positive and zero components as equilibrium $\mathcal E^*$, is called the \emph{positivity class} of $\mathcal E^*$.  Notice that the dimension of the subset $\Gamma_{\Omega}$ is $1+|\Omega_y|+|\Omega_z|$, where the notation $|\Omega_y|$ ($|\Omega_z|$) denotes the \emph{cardinality} of the set $\Omega_y$ ($\Omega_z$).  The equilibrium $\mathcal E^*$ must satisfy the following equations:
\begin{align}
\sum_{i\in\Omega_y} a_{ij} y_i^* &= \rho_j, \quad j\in \Omega_z  \notag \\
 \left( 1+ \sum_{i\in\Omega_y} \mathcal R_i y_i^*\right) &= \frac{1}{x^*}  \label{genEqcon}  \\
\sum_{j\in\Omega_z} a_{ij} Z_j^* &= \mathcal R_i x^* -1,  \quad i\in \Omega_y \notag
\end{align}
 We note that $\mathcal R_i>1,\ i\in \Omega_y$ must hold, even in the absence of CTL response.

Following Hofbauer and Sigmund \cite{hofbauer1998evolutionary}, we call an equilibrium $\mathcal E^*=(x^*,y^*,Z^*)$ of (\ref{ode3}) \emph{saturated} if the following holds when $\mathcal E^*$ has zero components:
\begin{align}
\mathcal R_i x^*-1-\sum_{j\in\Omega_z} a_{ij} Z_j^*\leq 0, \ \ \forall i\notin\Omega_y, \qquad \sum_{i\in\Omega_y} a_{ij} y_i^*-\rho_j \leq 0, \ \ \forall j\notin\Omega_z . \label{inequ}
\end{align}
Note that if $\mathcal E^*$ has all positive components, i.e. $\Omega_y=[1,m]$ and $\Omega_z=[1,n]$, the inequalities (\ref{inequ}) trivially hold.  Note also that each term in (\ref{inequ}) is an ``invasion'' eigenvalue of the Jacobian matrix evaluated at  $\mathcal E^*$ and thus a saturated equilibrium enjoys a weak stability
against invasion by missing species $i\notin \Omega_y,\ j\notin \Omega_z$. It immediately follows that a stable equilibrium must be saturated since the Jacobian cannot have a positive eigenvalue. As part of Theorem \ref{genThm} later in this section, we will conversely show that every saturated
equilibrium is stable. First, the following proposition states that there exists at least one saturated equilibrium.  For completeness we provide its proof in Appendix \ref{A1}, however we also note that this proposition follows directly from Theorem 2 in \cite{hofbauer1990index}.

\begin{proposition}\label{SatExist}
There exists a saturated equilibrium of system (\ref{ode3}).
\end{proposition}

More properties of relevant equilibria can be ascertained.  The following proposition states that any equilibria in the same positivity class must share the same value at $x^*$.

\begin{proposition} \label{xunique}
If $\mathcal E'=\left(x', y',  Z' \right)$ and $\mathcal E''=\left(x'', y'', Z'' \right)$ are both equilibria in the same positivity class, $\Gamma_{\Omega}$, then $x'=x''$.
\end{proposition}

\begin{proof}
Let $A'$ denote the submatrix of $A$ which contains only the rows in $\Omega_y$ and columns in $\Omega_z$.  Then the equilibrium conditions for $\mathcal E'$ can be rewritten as:
\begin{align*}
y'A'=\rho', \quad A'Z' = \frac{\vec{\mathcal R}'}{1+y'\vec{\mathcal R}'}-\vec{1},
\end{align*}
where $\rho'$ is the row vector with components $\rho_j$ where $j\in\Omega_z$, $\vec{\mathcal R}'$ is the column vector with components $\mathcal R_i$ where $i\in\Omega_y$ and $\vec{1}$ is the vector with all components one.  Since $\mathcal E''$ satisfies the same conditions, we obtain:
\begin{align*}
(y'-y'')A'&=\mathbf{0}, \qquad A'\left(Z'-Z''\right) = \left( \frac{(y''-y')\vec{\mathcal R}'}{(1+y'\vec{\mathcal R}')(1+y''\vec{\mathcal R}')}\right)\vec{\mathcal R}', \\
\Rightarrow 0&=(y'-y'')A'\left(Z'-Z''\right) =(y'-y'') \left( \frac{(y''-y')\vec{\mathcal R}'}{(1+y'\vec{\mathcal R}')(1+y''\vec{\mathcal R}')}\right)\vec{\mathcal R}' \\
\Rightarrow 0&=\left((y'-y'')\vec{\mathcal R}' \right)^2 \ \Rightarrow y'\vec{\mathcal R}' =y''\vec{\mathcal R}'  \\
\Rightarrow x'&=x''
\end{align*}
\end{proof}

 The previous proposition implies that if an equilibrium $\mathcal E^*=\left(x^*, y^*, Z^* \right)$ exists in positivity class $\Gamma_{\Omega}$, then any equilibrium $\mathcal E'=\left(x^*, y', Z' \right)$ belonging to $\Gamma_{\Omega}$ will satisfy the following Lotka-Volterra equilibria conditions within $\Gamma_{\Omega}$:
\begin{align}
\vec r\,'&+ B' \vec v'=\vec 0, \qquad \vec r\,'= \begin{pmatrix} \vec{\mathcal R}'x^*-\vec 1 \\ \vec \rho\,' \end{pmatrix}, \quad B'=\begin{pmatrix} 0 & -A' \\ (A')^T & 0 \end{pmatrix}, \label{LVequil}
\end{align}
where $A'$ is the submatrix of $A$ with rows in $\Omega_y$ and columns in $\Omega_z$, $\rho'$ is the row vector with components $\rho_j$ where $j\in\Omega_z$, $\vec{\mathcal R}'$ is the column vector with components $\mathcal R_i$ where $i\in\Omega_y$, and $\vec v\,'=\left(y',Z'\right)^T$.  Note that if the cardinality of $\Omega_y$ and $\Omega_z$ are equal ($|\Omega_y|=|\Omega_z|$) and $A'$ is non-singular, then clearly $\mathcal E^*$ is unique in its positivity class, $\Gamma_{\Omega}$.  More generally, if $B'$ is non-singular, then $\mathcal E^*$ is unique in $\Gamma_{\Omega}$.  The following proposition sharpens the condition for uniqueness of an equilibrium within a positivity class, and shows that in such equilibria the number of virus strains either is equal to or exactly one more than the number of immune responses.
\begin{proposition}\label{prop33}
Suppose the equilibrium $\mathcal E^*=\left(x^*, y^*, Z^* \right)$ exists in positivity class $\Gamma_{\Omega}$, where $\left(y^*,Z^*\right)$ satisfy the linear system of equations (\ref{LVequil}) and the cardinality of $\Omega_y$ and $\Omega_z$ are $|\Omega_y|=m'$ and $|\Omega_z|=n'$.  Then $\mathcal E^*$ is the unique equilibrium in $\Gamma_{\Omega}$, i.e. $\vec v=\left(y^*,Z^*\right)^T$ is the unique solution to (\ref{LVequil}), if and only if ${\rm Ker}(A')^T\cap \vec{\mathcal R}'^{\perp}=\left\{0\right\}$ and ${\rm Ker}(A')=\left\{0\right\}$.

Moreover, if $\mathcal E^*$ is the unique equilibrium in $\Gamma_{\Omega}$, then one of the following holds:
\begin{itemize}
\item[(i)] $m'=n'$, and $x^*=1/\left(1+(\vec{\rho}\,')^T(A')^{-1}\vec{\mathcal R}'\right)$.
\item[(ii)] $m'=n'+1$, and $x^*=\vec 1^{\,T} C^{-1}_{(n'+1)}$, where $C^{-1}_{(n'+1)}$ is the last column in the $(n'+1)\times (n'+1)$ matrix inverse of $C=\begin{pmatrix} A' & \vec{\mathcal R'}\end{pmatrix}^T$.
\end{itemize}
\end{proposition}
\begin{proof}
Suppose there exists equilibrium $\mathcal E^*$ in positivity class $\Gamma_{\Omega}$ and consider the corresponding $m'\times n'$ submatrix $A'$ from linear system (\ref{LVequil}), where $m'=|\Omega_y|$ and $n'=|\Omega_z|$.  Proposition \ref{xunique} shows that any two equilibria $\mathcal E'$ and $\mathcal E''$, contained in the same positivity class $\Gamma_{\Omega}$,  have equal $x^*$ components,  i.e. $x'=x''$.  Thus if these equilibria are distinct, then $y'\ne y''$ or $Z'\ne Z''$. In the proof of Proposition \ref{xunique} it is shown that
$$
y'-y'' \in \hbox{Ker} (A')^T\cap \vec{\mathcal R}'^{\perp} \ \hbox{and}\ Z'-Z''\in \hbox{Ker} A'.
$$
Therefore, the condition
$$
{\rm Ker}(A')^T\cap \vec{\mathcal R}'^{\perp}=\left\{0\right\}\ \hbox{and}\ {\rm Ker}(A')=\left\{0\right\}
$$
is equivalent to uniqueness of an equilibrium $\mathcal E^*$ in its positivity class $\Gamma_{\Omega}$.

Moreover $\hbox{Ker} (A')^T\cap \vec{\mathcal R}'^{\perp}=\{0\}$ if and only if the augmented $(n'+1)\times m'$ matrix $C$ consisting of
adding the final row $\vec{\mathcal R}'^T$ to $(A')^T$ has trivial kernel.  By the rank-nullity theorem, we obtain that $C$
has rank equal to $m'$. Since rank cannot exceed the number of rows, $m'\le n'+1$.
Applying the rank-nullity theorem to $A'$, gives rank equal to $n'$ and  $n'\le m'$. Thus $n'\le m'\le n'+1$.

In the case (i), $m'=n'$, the matrix $A$ is invertible and so from the equilibria equations (\ref{genEqcon}), we obtain $x^*=1/\left(1+(\vec \rho\,')^T(A')^{-1}\vec{\mathcal R}'\right)$.  Finally, consider case (ii), $m'=n'+1$.  Since $n'={\rm rank}(A')={\rm rank}\left((A')^T\right)$, by the rank-nullity theorem we obtain that ${\rm null}\left((A')^T\right)=1$.  We claim that $\hbox{Ker} (A')^T$ contains a vector $\vec w$ such that $\vec w^T \vec{\mathcal R}'=1$ and $\sum_i w_i =x^*$.  Since the matrix $C$ (defined in previous paragraph) has trivial Kernel, it is invertible.  Let $\vec w$ be the last column of $C^{-1}$.  Then it is not hard to see that  $(A')^T\vec w=\vec 0$ and $(\vec{\mathcal R}')^T\vec w=1$.  Furthermore
\begin{align*}
0&=\vec w^T A' Z' = \vec w^T \left(\vec{\mathcal R}'x^*-\vec{1}\right) = x^*\vec w^T \vec{\mathcal R}'-\sum_i w_i =x^* -\sum_i w_i.
\end{align*}
Thus $x^*=\sum_i w_i=\vec 1^T C^{-1}_{(n'+1)}$.
\end{proof}
Notice from the above proof that if an equilibrium $\mathcal E^*$ is not unique in its positivity class $\Gamma_{\Omega}$, then $\Gamma_{\Omega}$ contains an infinite number (a continuum) of equilibria.  Conversely, if $\mathcal E^*$ is unique in a positivity class $\Gamma_{\Omega}$ containing $n'$ (persistent) immune responses, then there are either (i) $n'$ virus strains or (ii) $n'+1$ virus strains in $\Gamma_{\Omega}$.

Additionally some results on existence of saturated equilibria in Lotka-Volterra systems can be recast in our setting.  For example, a sufficient condition for $\mathcal E^*$ to be saturated (and unique equilibrium in $\Gamma_{\Omega}$) is if the matrix $-B'$ in (\ref{LVequil}) is a $P$-matrix, i.e. all principal minors of $-B'$ are positive, by Theorem 15.4.5 in \cite{hofbauer1998evolutionary}.

In what follows, we will be interested in the global behavior of solutions to system (\ref{ode3}).  In doing so, we will determine which viral strains and immune responses uniformly persist \cite{thieme1993persistence} and which go extinct.  Define the system to be $\Omega_{yz}$ \emph{permanent} if
 \begin{align*}
 & \exists \ \epsilon,M>0 \ \text{and} \ T(\vec w_0) \ \text{such that} \  M>y_i(t),Z_j(t) >\epsilon, \ i \in\Omega_y, j\in\Omega_z, \ \forall t>T(\vec w_0), \ \text{and} \\
  & \lim_{t\rightarrow\infty}y_i(t),Z_j(t)= 0, \ i \notin\Omega_y, j\notin\Omega_z, \ \ \text{for every solution with initial condition} \ \vec w_0 \in \Omega.
   \end{align*}
     We will sometimes use the terminology that $y_i,Z_j \ i\in\Omega_y, j\in\Omega_z$ are uniformly persistent and $y_i,Z_j \rightarrow 0 \ i\notin\Omega_y, j\notin\Omega_z$ to signify the system being $\Omega_{yz}$ permanent.

 In the spirit of permanence as a sufficient condition for existence of a unique interior rest point in Lotka-Volterra systems \cite{hofbauer1998evolutionary}, we find the following proposition.
\begin{proposition}\label{permSat}
If the system is $\Omega_{yz}$ permanent, then there is a unique equilibrium $\mathcal E$ in the positivity class $\Gamma_{\Omega}$.\end{proposition}
\begin{proof}
By Theorem 6.2 in \cite{smith2011dynamical}, $\Omega_{yz}$ permanence implies that there exists an equilibrium in the positivity class $\Gamma_{\Omega}$.  Suppose by way of contradiction that there are two equilibria, $\mathcal E'$ and $\mathcal E''$, in $\Gamma_{\Omega}$.  By Proposition \ref{xunique}, $x'=x''$.  Since the remaining equilibria equations (\ref{genEqcon}) are linear, it can be shown that the line through $\mathcal E'$ and $\mathcal E''$ consist entirely of equilibria.  Then, we can find equilibria arbitrarily close to the boundary of $\Gamma_{\Omega}$.  This contradicts the fact that the system is $\Omega_{yz}$ permanent.
\end{proof}

 Now we state the main theorem of this section concerning the persistence of viral and immune variants of model \eqref{ode3}.  It builds on results by Browne \cite{browne2016global} concerning a special case of \eqref{ode3}, namely the case of perfectly nested subnetwork (see Fig. \ref{fig2c}) in the ``binary sequence'' system \eqref{odeS}.  In particular, the notion of saturated equilibria allows us to significantly extend persistence and stability results to the arbitrary interaction networks in general model \eqref{ode3}.

\begin{theorem}\label{genThm}
Suppose that $\mathcal E^*=\left(x^*,y^*, Z^* \right)$ is a non-negative equilibrium of system (\ref{ode3}) with positivity class $\Gamma_{\Omega}$.  Suppose further that $\mathcal E^*$ is saturated, i.e. the inequalities (\ref{inequ}) hold.
Then $\mathcal E^*$ is locally stable and $x(t)\rightarrow x^*$ as $t\rightarrow\infty$.

Furthermore, if $\mathcal E^*$ is the unique equilibrium in its positivity class $\Gamma_{\Omega}$ and the inequalities (\ref{inequ}) are strict, then  $y_i,Z_j\rightarrow 0$ for all $i\notin\Omega_y, j\notin\Omega_z$.  If $i\in \Omega_y$ and $a_{ij}=0 \ \forall j\in \Omega_z$, i.e. $\Lambda_i\cap \Omega_z=\emptyset$, then $y_i\rightarrow y_i^*$ and $x^*=1/\mathcal R_i$.  In addition, omega limit sets corresponding to positive initial conditions are contained in invariant orbits satisfying
\begin{align}
\sum_{i\in\Omega_y}\mathcal R_i y_i &= \sum_{i\in\Omega_y}\mathcal R_i y_i^*  \label{linEq}  \\
\dot y_i &= \gamma_i y_i\left(\sum_{j\in\Omega_z}a_{ij} \left( Z_j^*-Z_j \right)\right) \quad i\in \Omega_y \label{invSys}\\
  \dot Z_j &= \frac{\sigma_j}{\rho_j} Z_j\left(\sum_{i\in\Omega_y} a_{ij} \left( y_i - y_i^*\right) \right) \quad j\in\Omega_z, \notag
\end{align}
and for each $i\in\Omega_y, j\in\Omega_z$,  $y_i$ and $Z_j$ persist (the system is $\Omega_{yz}$ permanent) with asymptotic averages converging to equilibria values, i.e.
\begin{align*}
\lim_{t\rightarrow\infty} \frac{1}{t} \int\limits_0^t y_i(s) \, ds = y_i^*, \quad \lim_{t\rightarrow\infty} \frac{1}{t} \int\limits_0^t Z_j(s) \, ds = Z_j^*,
\end{align*}
In the case that there are less than or equal to two persistent viral strains with non-empty epitope sets (restricted to $\Omega_z$), i.e. $|\left\{ i\in\Omega_y: \Lambda_i\cap \Omega_z\neq \emptyset \right\}|\leq 2$, then $\mathcal E^*$ is globally asymptotically stable.
\end{theorem}

\begin{proof}
Consider the following Lyapunov function:
\begin{align*}
W(x,y,Z)&=x-x^*\ln\frac{x}{x^*} + \sum_{i=1}^m \frac{1}{\gamma_i}\left(y_i-y_i^*\ln\frac{y_i}{y_i^*} \right) + \sum_{j=1}^n \frac{\rho_j}{\sigma_j}\left(Z_j-Z_j^*\ln\frac{Z_j}{Z_j^*} \right) \\
&:=W_1+W_2+W_3,
\end{align*}
where the term with logarithm should be omitted if the corresponding coordinate in the particular equilibrium
is zero.  Then taking the time derivatives, we obtain
\begin{align*}
\dot W_1 &= 1- x -\frac{x^*}{x}+x^* + \sum_{i=1}^m \mathcal R_i y_i (x^*-x), \\
\dot W_2 &=  \sum_{i=1}^m \left(\mathcal R_i x-1-\sum_{j=1}^n a_{ij} Z_j \right) (y_i-y_i^*), \\
\dot W_3 &=  \sum_{j=1}^n \left(\sum_{i=1}^m a_{ij} y_i- \rho_j \right) \left(Z_j-Z_j^*\right)
\end{align*}
Thus
\begin{align*}
\dot W &= 1- x -\frac{x^*}{x}+x^* + \sum_{i=1}^m\left[ \mathcal R_i( y_i x^*-y_i^*x)-y_i+y_i^*\left(1+\sum_{j=1}^n a_{ij} Z_j  \right)\right]   \\
& \qquad \qquad \qquad \qquad \qquad \qquad -\sum_{j=1}^n\left[Z_j^*\sum_{i=1}^m a_{ij} y_i + \rho_j \left(Z_j-Z_j^*\right)\right]   \\
&= 1- x -\frac{x^*}{x}+x^* + \sum_{i=1}^m y_i \left( \mathcal R_i x^*-1 \right)+  \sum_{i\in\Omega_y} y_i^*\left(1+\sum_{j=1}^n a_{ij} Z_j -\mathcal R_i x \right) \\
& \qquad \qquad \qquad \qquad \qquad \qquad +   \sum_{j\in\Omega_z} Z_j^*\left(\rho_j-\sum_{i=1}^m a_{ij} y_i \right) -\sum_{j=1}^n \rho_j Z_j \\
&= 1- x -\frac{x^*}{x}+x^* + \sum_{i=1}^m y_i \left( \mathcal R_i x^*-1-\sum_{j\in\Omega_z} a_{ij} Z_j^* \right) +\sum_{j=1}^n Z_j\left(  \sum_{i\in\Omega_y} a_{ij} y_i^*-\rho_j \right)  \\
& \qquad \qquad \qquad \qquad \qquad \qquad+  \sum_{i\in\Omega_y} y_i^*\left(1 -\mathcal R_i x^*+ \sum_{j\in\Omega_z} a_{ij}  Z_j^* \right) +  \left(x^*-x \right)\sum_{i\in\Omega_y} \mathcal R_i y_i^*  \\
&= 1- x -\frac{x^*}{x}+x^* + \sum_{i=1}^m y_i \left( \mathcal R_i x^*-1-\sum_{j\in\Omega_z} a_{ij} Z_j^* \right) +\sum_{j=1}^n Z_j\left(  \sum_{i\in\Omega_y} a_{ij} y_i^*-\rho_j \right)  \\
& \qquad \qquad \qquad \qquad \qquad \qquad  +  \sum_{i\in\Omega_y} y_i^*\left(1 -\mathcal R_i x^*+ \sum_{j\in\Omega_z} a_{ij}  Z_j^* \right) +  \left(x^*-x \right)\frac{1-x^*}{x^*}  \\
&= -\frac{1}{x^*x}\left(x-x^*\right)^2 + \sum_{i\notin\Omega_y} y_i\left( \mathcal R_i x^*-1-\sum_{j\in\Omega_z} a_{ij} Z_j^* \right) +\sum_{j\notin\Omega_z} Z_j\left(  \sum_{i\in\Omega_y} a_{ij} y_i^*-\rho_j \right),
\end{align*}
where we make use of equilibrium conditions (\ref{genEqcon}) in the final line.  By the assumed inequalities (\ref{inequ}), we obtain that $\dot W \leq 0$, and thus $\dot W$ is a Lyapunov function at the equilibrium $\mathcal E^*$.  Noting that $\mathcal E^*$ is the unique minimizer of $W$, we obtain that $\mathcal E^*$ is (locally) stable.  Additionally, since $\dot W\leq 0$ and $W \rightarrow \infty$ as $y_i$, $z_j$ goes to $0$ or $\infty$ for $i\in\Omega_y, j\in\Omega_z$, we find that for any solution there exists $p,P>0$ such that $p\leq y_i,z_j\leq P$ for $i\in\Omega_y, j\in\Omega_z$.   Applying La Salle's Invariance principle (See e.g. Thm 2.6.1 in \cite{hofbauer1998evolutionary}), the $\omega$-limit set corresponding to any solution of (\ref{ode3}) with positive initial conditions is contained in the largest invariant set, $\mathcal L$, where $\dot W =0$.  Clearly $\dot W = 0 \Rightarrow x=x^*$, thus $x=x^*$ in $\mathcal L$.  This implies that $x(t)\rightarrow x^*$ as $t\rightarrow\infty$ for all solutions with positive initial conditions.  Furthermore if the inequalities (\ref{inequ}) are strict, then $y_i=0,z_j=0$ in $\mathcal L$ for $i\notin\Omega_y, j\notin\Omega_z$.  This implies that omega limit sets corresponding to positive initial conditions are contained in invariant orbits satisfying (\ref{invSys}) and (\ref{linEq}).  The differential equations in \eqref{invSys} can be integrated  to obtain asymptotic averages as follows:
 \begin{align}
&\frac{1}{ t\gamma_i}\ln y_i(t)=\frac{1}{ t}\int\limits_0^t \frac{\dot y_i}{\gamma_i y_i}ds = \frac{1}{ t}\int\limits_0^t\left[\sum_{j\in\Omega_z}a_{ij} \left( Z_j^*-Z_j \right)\right] ds \label{yiLog} \\
& \frac{1}{ t}\frac{\rho_j}{\sigma_j}\ln Z_j(t) = \frac{1}{ t}\int\limits_0^t \frac{\rho_j}{\sigma_j}\frac{\dot Z_j }{Z_j}dt= \frac{1}{ t}\int\limits_0^t \left(\sum_{i\in\Omega_y} a_{ij} \left( y_i - y_i^*\right) \right) ds \label{zjLog}
\end{align} Since $\mathcal E^*$ is the unique equilibrium in $\Gamma_{\Omega}$, utilizing a similar argument as \cite{hofbauer1998evolutionary}, we find that for each $i\in\Omega_y, j\in\Omega_z$, solutions $y_i,Z_j$ in the invariant set $\mathcal L$ satisfy:
\begin{align}
\lim_{t\rightarrow\infty} \frac{1}{t} \int\limits_0^t y_i(s) \, ds = y_i^*, \quad \lim_{t\rightarrow\infty} \frac{1}{t} \int\limits_0^t Z_j(s) \, ds = Z_j^* . \label{asymAvg}
\end{align}
  For any solution:
 $$\limsup_{t\rightarrow \infty}y_i(t) \geq y_i^* \quad \text{and} \quad \limsup_{t\rightarrow \infty}Z_j(t) \geq Z_j^*.$$
Therefore, $y_i$ and $Z_j$, $i\in\Omega_y, j\in\Omega_z$, are uniformly weakly persistent.  The key compactness hypotheses of Corollary 4.8 from \cite{smith2011dynamical}
are satisfied, and thus weak uniform persistence implies strong
uniform persistence for these variants.   In addition, note that if $i\in \Omega_y$ and $a_{ij}=0 \ \forall j\in \Omega_z$, i.e. $\Lambda_i\cap \Omega_z=\emptyset$, then $\dot y_i=0$ and therefore $y_i= y_i^*$ for this component on the invariant set $\mathcal L$ by (\ref{asymAvg}).

 Finally, we show that if $|\left\{ i\in\Omega_y: \Lambda_i\cap \Omega_z\neq \emptyset \right\}|\leq 2$, then $\mathcal E^*$ is globally asymptotically stable.   Without loss of generality suppose that $\left\{ i\in\Omega_y: \Lambda_i\cap \Omega_z\neq \emptyset \right\}=\left\{1, 2\right\}$.  Then $\mathcal R_1y_1+\mathcal R_2y_2=\mathcal R_1y_1^*+\mathcal R_2y_2^*$ by (\ref{linEq}) since $y_i=y_i^* \ i>2$ on the invariant set $\mathcal L$ for any other strains where $\Lambda_i\cap \Omega_z=\emptyset$.  Now clearly there exists time $t_1$ such that $y_1(t_1)=y_1^*$ by (\ref{asymAvg}), and this implies that $y_2(t_1)=y_2^*$ by the previous sentence.  Without loss of generality assume $t_1=0$.  Differentiating (\ref{linEq}) twice and evaluating at time $t=0$, we obtain:
\begin{align*}
0&=\mathcal R_1 \ddot y_1 +\mathcal R_2 \ddot y_2 \\
&= \sum_{i=1}^2\mathcal R_i\gamma_i\left[ \gamma_i \left(\sum_{j\in\Omega_z}\left( Z_j^*-Z_j(0) \right)\right)^2 -  y_i(0)\left(\sum_{j\in\Omega_z}a_{ij}\frac{\sigma_j}{\rho_j}Z_j\sum_{k\in\Omega_y} a_{kj}\left(y_k(0)-y_k^*\right) \right)\right] \\
0&=  \sum_{i=1}^2\mathcal R_i\gamma_i^2\left(\sum_{j\in\Omega_z}\left( Z_j^*-Z_j(0) \right)\right)^2,
 \end{align*}
where the last equality comes from the fact that $y_k(0)=y_k^*$ for $k=1,2$, and $a_{kj}=0$ for $k>2$.  The only way the above equation can be satisfied is if $Z_j(0)=Z_j^*$ for all $j$.  Then $y_i(0)=y_i^*, \ Z_j(0)=Z_j^*$, $i\in\Omega_y, j\in\Omega_z$.  Thus $\mathcal E^*$ is globally asymptotically stable in this case.

\end{proof}

More general results can be obtained in special cases, in particular, the inequalities (\ref{inequ}) need not be strict for persistence results in certain cases discussed in Section \ref{Sec5}.  However, the global convergence to persistent variants is still an open question when there are more than two persistent immune responses.  Note that the proof for global stability of ``strictly saturated'' equilibria $\mathcal E^*$ with $|\left\{ i\in\Omega_y: \Lambda_i\cap \Omega_z\neq \emptyset \right\}|\leq 2$ does not extend to higher dimensional equilibria.  Our numerical simulations that we have conducted support the global stability of $\mathcal E^*$ in general.  We conjecture the following:
\begin{conjecture} \label{conj}
An equilibrium $\mathcal E^*$ of (\ref{ode3}), which is unique in its positivity class $\Gamma_{\Omega}$ and satisfies inequalities (\ref{genEqcon}) strictly, is globally asymptotically stable for positive initial conditions (regardless of the dimension of $\Gamma_{\Omega}$).
\end{conjecture}

We briefly discuss Theorem \ref{genThm} in the context of Lotka-Volterra (L-V) systems.  L-V systems take the general form
\begin{align}
\frac{dv_i}{dt}=v_i\left(r_i+b_{ij}v_j\right) \quad i=1,\dots,N . \label{LV}
\end{align}
The multi-trophic version of model \eqref{LV} with $m$ prey and $n$ predators may be set up with similar prey-predator interaction rates as our model (see \eqref{LVequil}).  The main difference between our chemostat-type system \eqref{ode3} and the L-V model is the explicit inclusion of the resource, $x$, modulating competition in the ``prey'' species (virus strains) in \eqref{ode3}, as opposed to possible logistic competition terms, $b_{ij}<0, \ 1\leq i,j\leq m$  in L-V models.  Theorem \ref{genThm} reduces the dynamics of our model \eqref{ode3} (in the case of a ``$\Gamma_{\Omega}$-unique'' strictly saturated equilibrium) to the invariant L-V system \eqref{invSys} with $|\Omega_y|$ prey and $|\Omega_z|$ predators subject to the constraint \eqref{linEq} on prey.   Global stability of equilibria in L-V systems have been studied extensively (see e.g. \cite{takeuchi1996global}).  For instance, global stability of distinct equilibria was proved in the case of $m=n$ ($N=2n$ in \eqref{LV}) and uniform competition coefficient for prey ($b_{ij}=b \ 1\leq i,j\leq n$) for a perfectly nested network \cite{korytowski2017persistence}. For general LV systems, Goh proved that a strictly saturated equilibrium $\mathbf v^*$ is globally asymptotically stable for system \eqref{LV} if there exists a positive diagonal matrix $D$ such that $DB+B^TD$ is negative definite (except possibly at equilibrium $\mathbf v^*$) \cite{Goh}.   Unfortunately none of these global stability results on L-V models apply to the invariant system \eqref{invSys}, in part due to the lack of intraspecific competition coefficients in \eqref{invSys}.  Indeed the simple one predator-one prey L-V planar system of the form \eqref{invSys} displays oscillatory solutions in absence of the additional constraint \eqref{linEq}.  The constraint \eqref{linEq} seems to induce saturated equilibria to be attractors, leading to our conjecture (Conjecture \ref{conj}) stating that in Theorem \ref{genThm}, variants $y_i,z_j \ i\in\Omega_y,j\in\Omega_z$ are not only uniformly persistent but also globally converge to their equilibrium values.

%Can we get more?
% permanent implies unique equilibrium?
%Conditions ensuring saturated equilibria?

\section{Special Cases of Multi-Epitope Model}\label{Sec5}
In this section we consider the ``binary sequence'' case of model \eqref{ode3}, which was introduced in Section \ref{BMcase}.  Recall that the assumption \eqref{assumeaj} leads to the rescaled and simplified system \eqref{odeS}.  Biologically, we are considering the situation where $n$ immune response populations $z_j$ each target the corresponding epitope $j$ in the virus strains at a rate solely dependent on the allele type of epitope $j$; (0) wild-type or (1) mutated form conferring full resistance to $z_j$.  The avidity of immune response $z_j$ and (wild-type) epitope $j$ is described by the immune reproduction number $\mathcal I_j$ given by \eqref{IR0}.
The resulting model \eqref{odeS} has $m=2^n$ potential viral strains distinguished by their (basic) reproduction number, $\mathcal R_i$, and epitope set, $\Lambda_i$.    Recall that in model \eqref{odeS} the viral strains can be viewed in a mutational pathway network, where each strain, $y_i \ i=1,\dots,2^n$, is a vertex in the \emph{hypercube graph}, $Q_n$, corresponding to the strain's epitope set, $\Lambda_i$, represented as a binary sequence of length $n$, $\mathbf i=(i_1,\dots,i_n)$.  See Section \ref{BMcase} for relevant explanation and definitions.

We can establish restrictions on the positivity class of feasible equilibria in model (\ref{odeS}) based on graph-theoretic considerations of the viral strains viewed as vertices in the hypercube graph.  Indeed, we will show that equilibria with persistent viral strains forming a \emph{cycle} of order $2^j$ can only occur in the degenerate cases.  A \emph{(simple) cycle} is defined as a closed path in the hypercube graph from containing no other repetition of vertices other than the starting and ending vertex.  
 For $2\leq j\leq n$, there are $2^{n-j}$ disjoint $2^{j}$-cycles which cover the vertices of the hypercube graph $Q_n$.  In particular, it is well-known that there is a \emph{Hamiltonian} cycle covering $Q_n$, i.e. a cycle that includes every edge in the graph.
  The following proposition concerns feasibility of equilibria with cycles in the viral hypercube graph (here the important aspect is the set of vertices on the cycle).
\begin{proposition}\label{CycleProp}
 For $n\geq 2$, consider model (\ref{odeS}).  Let $2\leq j\leq n$ and $\ell=2^{j}$.  Suppose there is a simple $\ell$-cycle in the representative hypercube graph $Q_n$: $y_{k_1}\leftrightarrow y_{k_2}\leftrightarrow \dots\leftrightarrow y_{k_{\ell}}\leftrightarrow y_{k_1}$.  If $\sum_{i=1}^{\ell} (-1)^i \mathcal R_{k_i}=\sum_{i=1}^{\ell} (-1)^{d(y_{k_i},y_w)} \mathcal R_{k_i} \neq 0$, then there does not exist an equilibrium with $y^*_{k_1},y^*_{k_2}, \dots, y^*_{k_{\ell}}>0$.
\end{proposition}
\begin{proof}
Fix $2\leq j\leq n$ and let $\ell=2^{j}$.  Without loss of generality, consider an $\ell$-cycle with strains labeled as $y_1\leftrightarrow \dots\leftrightarrow y_{\ell}\leftrightarrow y_1$.  The simple cycle of length $\ell$ can be seen as an embedded hypercube graph on $2^j$ vertices, $Q_j$, corresponding to binary strings where $n-j$ slots are fixed to be all zeros or all ones.  Without loss of generality suppose that these fixed slots are indices $[j+1,n]$ corresponding to $z_{j+1},\dots z_n$, i.e. either $[j+1,n]\subseteq \Lambda_i$ or $[j+1,n]\cap \Lambda_i=\emptyset$ for all $i=1,\dots,\ell$.  The viral strains can be grouped into the classes $\mathcal C_k$, $k=0,\dots j$, where $k$ is the number of mutated alleles (1) appearing in the epitope set $[1,j]$, i.e. $y_i\in \mathcal C_k$ if $k=j-|\Lambda_i\cap [1,j]|$.    Notice that in $\mathcal C_k$ there are exactly $\binom{j-1}{k}$ viral strains, $y_i$, with epitope 1 being wild-type allele (0), i.e. $1\in\Lambda_i$.  The analogous statement holds for any epitope in $[1,j]$.  Thus, for equilibria of \eqref{odeS}, we find that
\begin{align*}
0&=\sum_{i=1}^{\ell} (-1)^{i} \frac{\dot y_i}{\gamma_i y_i^*} \\ 
&=x^*\sum_{i=1}^{\ell} (-1)^{i} \mathcal R_i -\sum_{i=1}^{\ell} (-1)^{i}- \sum_{k=j+1}^{n}z_k^*\sum_{i=1}^{\ell} (-1)^{i}  - \sum_{k=1}^{j}z_k^* \sum_{i=0}^{j-1} \binom{j-1}{i} (-1)^i   \\
&= x^*\sum_{i=1}^{\ell} (-1)^{i} \mathcal R_i .
\end{align*}
Note here that $\sum_{i=1}^{\ell} (-1)^i \mathcal R_{i}=\sum_{i=1}^{\ell} (-1)^{d(y_i,y_w)} \mathcal R_{i}$ since we must traverse the simple cycle through ``distance-one'' single epitope mutations.
\end{proof}

A couple remarks about the above proposition are in order.  First, we highlight the of case $j=2$ of Proposition \ref{CycleProp}, which states the generic non-existence of equilibria with persistent viral strains forming a 4-cycle in the associated hypercube graph.  In particular, if viral strains $y_{k_1},\dots,y_{k_4}$ form a cycle, then there exists an equilibrium $\mathcal E^*$ with $y_{k_1}^*,\dots,y_{k_4}^*>0$ only if $\sum_{i=1}^{4} (-1)^{p_i} \mathcal R_{k_i}=0$ where $p_i=d(y_{k_i},y_w)$.  In Section \ref{TwoEpitope}, the degeneracy of ``four-cycle equilibria'' will be detailed further for the case of $n=2$ epitopes.   In general, for $n\geq 2$, there are $2^{n-2}$ disjoint 4-cycles which cover the vertices of $Q_n$, and their unions form larger cycles with order as powers of two.  Each of these cycles of order $2^j, \ j\geq 2,$ can be seen as an embedded $j$-dimensional hypercube.  Thus the proposition establishes the degeneracy of any equilibrium with positive viral components forming an embedded hypercube subgraph in the associated hypercube $Q_n$, which suggests that these combinations of viral strains generally do not persist together in model (\ref{odeS}).

In following subsections, we analyze a few special cases of the multi-epitope model \eqref{odeS} where stable equilibria can be sharply characterized by quantities derived from the parameters using Theorem \ref{genThm}.  We assume without loss of generality that immune responses $z_1,\dots, z_n$ are ordered according to an immunodominance hierarchy:
\begin{align}\label{immunodom}
s_1\leq s_2 \leq \dots \leq s_n, \ \ \text{i.e.}  \ \  \mathcal I_1\geq \mathcal I_2 \geq \dots \geq \mathcal I_n.
\end{align}
Section \ref{NestedSec} summarizes previously published analysis of model \eqref{odeS} when the network is constrained to be perfectly nested, and Section \ref{ssSec} contains prior results, along with a new proposition, about the strain-specific (one-to-one) network.  These cases provide nice applications of our general Theorem \ref{genThm}, and also are important for the analysis of the full network.  We consider the full network with $m=2^n$ virus strains in section \ref{TwoEpitope} for $n=2$ epitopes, and in section \ref{EqFitSec} for arbitrary $n$ in the special case of equal fitness costs for each epitope mutation.

\subsection{Nested Network} \label{NestedSec}
 While the virus-immune epitope interaction network generally can be quite complex, patterns of viral escape and dynamic immunodominance hierarchies often emerge.  In observations of HIV infection, the initial CTL response occurs at a few immunodominant epitopes and is followed by viral mutations at these epitopes conferring resistance, along with a fall in these specific CTLs and rise in subdominant CTLs  \cite{liu2013vertical}.  This pattern continues, albeit at diminishing rates as time proceeds, resulting in viral strains with resistance at multiple epitopes and corresponding fitness costs, along with subdominant CTLs of increasing breadth.  An idealized description of this process is a \emph{perfectly nested network}, where resistance to multiple epitopes is built sequentially according to the immunodominance hierarchy.   Nested networks have been of recent interest in explaining the biodiversity and structure of bacteria-phage communities \cite{jover2013mechanisms,korytowski2015nested,weitz2013phage}, and there is some evidence that nestedness is a feature of HIV-CTL dynamics \cite{kessinger2015inferring,liu2013vertical,vanDeutekom}.   Below we summarize recent work \cite{browne2016global}, characterizing the stability of equilibria and uniform persistence or extinction of the populations for system (\ref{odeS}) in the case of a perfectly nested network.

 The perfectly nested network consists of $n$ epitope specific CTLs, $z_1,\dots,z_n$, and $m=n+1$ virus strains $y_1,\dots,y_{n+1}$ where the epitope set of $y_i$ is $\Lambda_i=\left\{ i,\dots, n \right\}$ (having escaped immune responses $z_1,\dots,z_{i-1}$). The equations are
  \begin{align}
\dot x &= 1-x- x\sum_{i=1}^n \mathcal R_i y_i, \quad \dot y_i = \gamma_i y_i\left(\mathcal R_i x -1 -\sum_{j\ge i}  z_j \right), \quad \dot y_{n+1} = \gamma_{n+1} y_{n+1}\left(\mathcal R_{n+1} x -1\right), \notag \\
& \qquad\qquad  \dot z_i = \frac{\sigma_i}{s_i} z_i\left(\sum_{j\le i} y_j - s_i \right), \quad \text{where}\ i=1,\dots,n.  \label{odeNN}
\end{align}

 As before we assume a fitness cost for each mutation, and here we also assume a strict immunodominance:  $$ \mathcal R_1>\mathcal R_2>\dots > \mathcal R_{n+1} \quad \text{and} \quad \mathcal I_1>\mathcal I_2>\dots > \mathcal I_n .$$
 Out of a multitude of non-negative equilibria ($>2^n$), there are $2n+2$ which can be stable, and their stability depends upon quantities derived from parameters as follows.  For $k\geq 1$ define:
 \begin{align}
 \mathcal Q_{k}= \mathcal Q_{k-1}+ (s_{k}-s_{k -1})\mathcal R_{k}, \quad\text{where} \quad  \mathcal Q_0=1, s_0=0, s_k=1/\mathcal I_k. \label{nestedQ}
  \end{align}
Then, for each $k\in[0,n]$, define the following equilibria:
 \begin{align}
 \widetilde{\mathcal E}_{k+1}   = (\widetilde x,\widetilde y,\widetilde z),  \qquad  & \widetilde x=\frac{1}{\mathcal R_{k+1}}, \ \widetilde y_i=s_i-s_{i-1}, \ \widetilde z_i=\frac{\mathcal R_i-\mathcal R_{i+1}}{\mathcal R_{k+1}} \ \  \ \text{for} \ \ 1\leq i\leq k, \label{equilib1} \\  & \widetilde y_{k+1}=1-\frac{ \mathcal Q_{k}}{\mathcal R_{k+1}}, \ \widetilde z_{k+1}=0, \quad \widetilde y_i=\widetilde z_i=0 \  \ \ \text{for} \ \ k+1<i\leq n \notag \\
 \bar{\mathcal E}_k = (\bar x,\bar y,\bar z),  \qquad &  \bar x=\frac{1}{\mathcal Q_k}, \ \bar y_i=s_i-s_{i-1}, \ \bar z_i=\frac{\mathcal R_i-\mathcal R_{i+1}}{\mathcal Q_k} \ \  \ \text{for} \ \ 1\leq i<k,  \label{equilib2}  \\ &  \bar y_k=s_k-s_{k-1}, \ \bar z_k=\frac{\mathcal R_k}{\mathcal Q_k}-1, \quad  \bar y_i=\bar z_i=0 \  \ \ \text{for} \ \ k<i\leq n, \ \bar y_{k+1}=0  \notag
  \end{align}
Equilibrium $\widetilde{\mathcal E}_{k+1}$ represents the appearance of escape mutant $y_{k+1}$ from equilibrium  $\bar{\mathcal E}_k$.

 The main result of \cite{browne2016global} is as follows:
\begin{theorem}[\cite{browne2016global}]\label{NestedThm}  If $\mathcal R_1> \mathcal Q_1$, let $k$ be the largest integer in $[1,n]$ such that $\mathcal R_k > \mathcal Q_k$, otherwise let $k=0$. If $\mathcal R_{k+1}\leq \mathcal Q_{k}$, then for $1\leq i\leq k$, $y_i,z_i$ are uniformly persistent (if $\mathcal R_{k+1}>\mathcal Q_k$, $y_{k+1}$ persists), and the other variants globally converge to zero, $x(t)\rightarrow 1/\mathcal Q_k$ (if $\mathcal R_{k+1}>\mathcal Q_k$, $x(t)\rightarrow 1/\mathcal R_{k+1}$ and $y_{k+1}(t)\rightarrow 1-\frac{ \mathcal Q_{k}}{\mathcal R_{k+1}}$) as $t\rightarrow\infty$.  Additionally, the corresponding equilibria ($\widetilde{\mathcal E}_{k+1}$ or $\overline{\mathcal E}_{k}$) are locally stable (globally asymptotically stable when $k=0,1,2$),  asymptotic averages converge to equilibria, and the global attractor satisfies (\ref{linEq}) and (\ref{invSys}).
\end{theorem}
Note that the proof can be obtained through application of Theorem \ref{genThm}.

 Theorem \ref{NestedThm} suggests a stable diverse set of viral strains and immune response which can be built up by the nested accumulation of epitope resistance and rise of subdominant CTLs.  The diversity achieved depends upon potential breadth $n$ and ``immune invasion'' number at epitope $k\leq n$, $\mathcal R_k/\mathcal Q_k$, which depends upon the strengths of CTL directed at the $k$ epitopes in immunodominance hierarchy and the viral fitness costs of $k$ sequential mutations, along with initial fitness $\mathcal R_1$.  Observe that since $\mathcal R_k$ decreases and $\mathcal Q_k$ increases with breadth $k$, Theorem \ref{NestedThm} implies exclusion of $y_{k+1}$ is more likely as the breadth increases.  Additionally, it is shown in  \cite{browne2016global} that the rate of $y_{k+1}$ invasion decreases as the breadth $k$ increases, which is consistent with several studies showing rate of HIV viral escape from CTL responses slows down after acute infection, along with relatively few escapes \cite{asquith2006inefficient,Vitaly1}.

 Overall, the analysis in the case of the nested network confirms some patterns of multi-epitope viral escape and reinforces the importance of strong immune responses directed at conserved epitopes (high fitness cost for resistance) in order to control HIV with CTL response.  However, constraining multi-epitope resistance to be built in a nested fashion leaves out other potential mutational pathways.  
 
 The question remains, with $n$ epitopes targeted by distinct immune responses, what are \emph{all} of the potential escape patterns and stable equilibria?  We will answer this question for the simplest case of $n=2$ epitopes in Section \ref{TwoEpitope}, but first we consider the ``strain-specific'' (or one-to-one) subnetwork.

\subsection{Strain-Specific network}\label{ssSec}
Consider the case where $m=n+1$ with $\Lambda_i=\left\{i\right\}$ and $\Lambda_{n+1}=\emptyset$ in model (\ref{odeS}), i.e.  in system (\ref{ode3}), $A$ is a $n+1\times n$ matrix comprised of the diagonal matrix ${\rm diag}\left(a_1,\dots, a_n\right)$ and a row of zeros. This particular assumption of a ``one-to-one'' interaction network, where each immune response population attacks a unique specific viral strain, has been considered in \cite{wolkowicz1989successful,korytowski2015nested,bobko2015singularly}.  In this case, model (\ref{odeS}) reduces to the following $n+1$-strain and $n$-immune variant model:
\begin{align}
\dot x &= 1-x- x\sum_{i=1}^n \mathcal R_i y_i, \quad \dot y_i = \gamma_i y_i\left(\mathcal R_i x -1 -z_i \right), \quad \dot y_{n+1} = \gamma_{n+1} y_{n+1}\left(\mathcal R_{n+1} x -1\right), \notag \\
& \qquad\qquad  \dot z_i = \frac{\sigma_i}{s_i} z_i\left(y_i -s_i \right), \quad \text{where} \ i=1,\dots,n.  \label{odeSS}
\end{align}
 Note that virus strain $y_{n+1}$ is immune to all immune cells.  First we summarize previous results on system \eqref{odeSS}, which all can be derived from our general analysis in Section \ref{Sec3}.  For these results, we impose the additional assumption that $\mathcal R_i$ are decreasing with $i$, along with our immundominance hierarchy (\ref{immunodom}), to avoid degeneracy,\begin{equation}\label{Rorder}
\mathcal R_1> \mathcal R_2 > \cdots > \mathcal R_{n+1}.
\end{equation}
 For each $k\in[1,n]$, an equilibrium $\mathcal E^*$ with persistent variant sets $\Omega_y=\Omega_z=[1,k]$, must have component $x^*=1/\mathcal P_k$ by equation \eqref{genEqcon} (or by Proposition \ref{prop33}), where:
 \begin{align}
 \mathcal P_{k}= \mathcal P_{k-1}+ s_{k}\mathcal R_{k}, \quad\text{where} \quad  \mathcal P_0=1, s_k=1/\mathcal I_k. \label{ssP}
  \end{align}
More generally we can define $\mathcal P_{\mathcal J}=1+\sum_{i\in\mathcal J}\mathcal R_is_i$ for any subset $\mathcal J \subseteq [1,n]$. Then for each $k\in[0,n]$, the following equilibria are found:
 \begin{align}
\mathcal E^{\ddagger}_{k+1}   = (x^{\ddagger},y^{\ddagger},z^{\ddagger}),  \qquad  & x^{\ddagger}=\frac{1}{\mathcal R_{k+1}}, \quad y_i^{\ddagger}=s_i, \quad z_i^{\ddagger}=\frac{\mathcal R_i}{\mathcal R_{k+1}}-1, \quad  i=1,\dots, k, \label{SSeq1} \\ & \quad y_{k+1}^{\ddagger}=1-\frac{\mathcal P_k}{\mathcal R_{k+1}}, \quad  z_{k+1}^{\ddagger}= 0, \quad y_i^{\ddagger}= z_i^{\ddagger}= 0, \quad i>k+1, \notag \\
\mathcal E^{\dagger}_{k}   = (x^{\dagger},y^{\dagger},z^{\dagger}),\qquad &
x^{\dagger}=\frac{1}{\mathcal P_k},  \quad y_i^{\dagger}=s_i,  \quad z_i^{\dagger}=\frac{\mathcal R_i}{\mathcal P_k}-1, \quad  i=1,\dots, k, \label{SSeq2} \\ & \quad  \quad y_i^{\dagger}= z_i^{\dagger}= 0, i\geq k+1,\quad y_i^{\dagger}= z_i^{\dagger}= 0, \quad i>k, \notag
  \end{align}

Let
$k\in [1,n+1]$ be maximal such that
$$
\mathcal R_k>1+\sum_{i=1}^{k-1}\mathcal R_is_i
$$
where the sum on the right vanishes if $k=1$. There are two cases depending on whether (for $k<n+1$)
$$
\mathcal R_k>1+\sum_{i=1}^{k}\mathcal R_is_i
$$
or not.  If the inequality holds (for $k<n+1$), then there is a unique saturated equilibrium, $\mathcal E^{\dagger}_k$, with $\Omega_z=\Omega_y=[1,k]$
and if it does not or $k=n+1$, then there is a unique saturated equilibrium, $\mathcal E^{\ddagger}_{k}$, with $\Omega_z=[1,k-1],\Omega_y=[1,k]$.  The inequalities (\ref{inequ}) are strict in either case so the system is $\Omega_{yz}$ permanent
and $y_i,z_j\to 0$ for all $i\notin \Omega_y,j\notin \Omega_z$, by Theorem \ref{genThm}.

In fact, much more can be said about the asymptotic behavior of solutions.  In particular, the other conclusions of Theorem \ref{genThm} imply that the components $y_i,z_j$ where $i\notin\Omega_y, j\notin\Omega_z$ converge to zero, the asymptotic means of the persistent variants converge to equilibria values, and the global attractor satisfies (\ref{invSys}) which consists of $k$ or $k-1$ distinct planar Lotka-Volterra predator-prey differential equations for $(y_i,z_i)$ constrained by the relation $\sum_i \mathcal R_i y_i = \sum_i \mathcal R_is_i$.  This fact was utilized in a chemostat model by Wolkowicz \cite{wolkowicz1989successful} to prove that a unique saturated equilibrium is globally asymptotically stable for $k=1,2$ (including when $\Omega_z=\left\{1,2\right\},\Omega_y=\left\{1,2,3\right\}$).  These results are summarized in the following theorem.

  \begin{theorem}[\cite{wolkowicz1989successful}]\label{SSThm}  If $\mathcal R_1> \mathcal P_1$, let $k$ be the largest integer in $[1,n]$ such that $\mathcal R_k > \mathcal P_k$, otherwise let $k=0$. If $\mathcal R_{k+1}\leq \mathcal P_{k}$, then for $1\leq i\leq k$, $y_i,z_i$ are uniformly persistent (if $\mathcal R_{k+1}>\mathcal P_k$, $y_{k+1}$ also persists), and the other variants globally converge to zero, $x(t)\rightarrow 1/\mathcal P_k$ (if $\mathcal R_{k+1}>\mathcal P_k$, $x(t)\rightarrow 1/\mathcal P_{k+1}$ and $y_{k+1}(t)\rightarrow 1-\frac{ \mathcal P_{k}}{\mathcal R_{k+1}}$) as $t\rightarrow\infty$.  Additionally, the corresponding equilibria ($\mathcal E^{\dagger}_{k} $ or $\mathcal E^{\ddagger}_{k+1}$) are locally stable (globally asymptotically stable when $k=0,1,2$),  asymptotic averages converge to equilibria, and the global attractor satisfies (\ref{linEq}) and (\ref{invSys}).
\end{theorem}

If the assumption of strictly decreasing reproduction numbers are relaxed, then the strain-specific system \eqref{odeSS} can have multiple degenerate saturated equilibria.  However the full hypercube network for $n$ epitopes containing $2^n$ virus strains (model \eqref{odeS}) allows us to relax this particular assumption on reproduction numbers,  avoid any degeneracy, and explicitly identify two possible saturated ``strain-specific equilibria''.   Essentially, $\mathcal E^{\dagger}_n$ and  $\mathcal E^{\ddagger}_{n+1}$ are the only strain-specific equilibria which can be saturated in the full hypercube network.  The proposition below (proved in Appendix \ref{ssA}) contains this detailed result on strain-specific equilibria.
\begin{proposition}\label{ssF}
Consider system \eqref{odeS} on the full network with $n$ epitopes ($m=2^n$ virus strains) under the assumption of mutational fitness costs \eqref{fitnesscost} and the viral strains, $y_i$, ordered so that $\Lambda_i=\left\{i\right\}$ for $i=1,\dots,n$ and $\Lambda_{n+1}=\emptyset$.  Let $\mathcal P_n$ be defined by \eqref{ssP}.  Then an equilibrium $\mathcal E^*$ with a strain-specific subgraph, i.e. $\Omega_y\subseteq [1,n+1]$, is saturated if and only if one of the following holds:
\begin{itemize}
\item[i.]  $\mathcal R_{n+1}\leq \mathcal P_n$ and $\left( |\Lambda_{\ell}| - 1 \right) \mathcal P_n + \mathcal R_{\ell} \leq  \sum\limits_{i\in\Lambda_{\ell}} \mathcal R_i  \quad \forall \ell \in [n+2, 2^n]$, in which case $\Omega_y=\Omega_z= [1,n]$.
\item[ii.]  $\mathcal R_{n+1}> \mathcal P_n$ and $\left( |\Lambda_{\ell}| - 1 \right) \mathcal R_{n+1} + \mathcal R_{\ell} \leq  \sum\limits_{i\in\Lambda_{\ell}} \mathcal R_i  \quad \forall \ell \in [n+2, 2^n]$, in which case $\Omega_y=[1,n+1]$ and $\Omega_z= [1,n]$.
\end{itemize}  
\end{proposition}
Therefore by applying Theorem \ref{genThm} to the above Proposition \ref{ssF}, we obtain conditions for the stability of either (i) $\mathcal E^{\dagger}_n$ (or (ii) $\mathcal E^{\ddagger}_{n+1}$), and the corresponding uniform persistence of $y_1,\dots,y_n$ (and $y_{n+1}$ in case of (ii)).  Furthermore these are the only potential stable equilibria with persistent strain set contained in the strain-specific subnetwork ($\Omega_y\subseteq [1,n+1]$) for system \eqref{odeS}.

\subsection{Dynamics for full network on $n=2$ epitopes} \label{TwoEpitope}
If two epitopes are concurrently targeted by two distinct specific immune responses, which escape pathway will the virus follow and what mutant strains persist?  In the nested network, we assumed that the virus escaped the most immunodominant response.  However, in general, both \textit{CTL pressure} and \textit{virus fitness cost} determine selective advantage of a resistant mutant.  For a single epitope, an escape mutant $y_2$ invades the wild-type $y_1$, if its reproductive number is large enough given the CTL pressure on $y_1$.   In particular, $y_2$ persists if $(f-s_1)\mathcal R_1>1$, where $\mathcal R_2=f\mathcal R_1$ and $s_1=1/\mathcal I_1$ with $f$ the fitness proportion of the wild-type reproductive number $\mathcal R_1$ and $\mathcal I_1$ the immune response reproduction number (see \cite{pruss2008global} or the  case $n=1$ in Sections \ref{NestedSec} or \ref{ssSec}).   For the general case of $n=2$ epitopes, although the situation is fundamentally more complex, we will sharply characterize the dynamics in this section.  The results suggest that immunodominance may play a larger role than viral fitness in determining the structure of the persistent virus-immune network.

The full network for $n=2$ epitopes (shown in Fig. \ref{fig2d}) consists of 2 CTL populations, $z_1$ and $z_2$, and $m=4$ virus strains, $y_i \ i=1,\dots,4$, each with an associated binary string describing their resistance profile; $y_1\sim 00$, $y_2\sim 10$, $y_3\sim 01$, $y_4\sim 11$.  Recall that a 0 in the $j^{th}$ bit of the binary string signifies susceptible to $z_j$, whereas 1 signifies resistance; for example $y_2 \ (10)$ is resistant to $z_1$ but susceptible to $z_2$.  We assume that $z_1$ is strictly immunodominant, i.e. $\mathcal I_1>\mathcal I_2$ or $s_1<s_2$.  Each epitope escape comes with a fitness cost as before, therefore the viral reproduction numbers satisfy $\mathcal R_1 > \max(\mathcal R_2, \mathcal R_3)\geq \min(\mathcal R_2, \mathcal R_3)> \mathcal R_4$.  We note that the fitness cost for resistance to $z_1$ may be greater than or equal to the fitness cost to $z_2$, in other words the fitness of mutant $y_2$ is less than or equal to $y_3$ ($\mathcal R_2\leq \mathcal R_3$).   Conversely, it may be the case that resistance to the dominant immune response, $z_1$, comes at less cost than resistance to the weaker response ($\mathcal R_2 > \mathcal R_3$).  Our underlying assumptions are summarized below:
\begin{align}
\mathcal I_1&>\mathcal I_2 \ \ (s_1<s_2), \qquad \mathcal R_1 > \max(\mathcal R_2, \mathcal R_3)\geq \min(\mathcal R_2, \mathcal R_3)> \mathcal R_4 \label{conditions}
\end{align}
For clarity, we write the 7 equations in model (\ref{odeS}) for this case $n=2$ with the chosen index notation:
\begin{align}
\dot x &= 1-x- x\sum_{i=1}^4 \mathcal R_i y_i, \quad
\dot y_1 = \gamma_1 y_1\left(\mathcal R_1 x -1 -(z_1+ z_2) \right), \quad
\dot y_2 = \gamma_2 y_2\left(\mathcal R_2 x -1 -z_2 \right), \notag \\ 
\dot y_3 &= \gamma_3 y_3\left(\mathcal R_3 x -1 - z_1 \right), \quad
\dot y_4 = \gamma_4 y_4\left(\mathcal R_4 x -1 \right), \label{ode4}  \\
\dot z_1 & = \frac{\sigma_1}{s_1} z_1\left(y_1 + y_3-s_1 \right), \quad
  \dot z_2 = \frac{\sigma_2}{s_2} z_2\left(y_1 + y_2-s_2 \right), \notag
  \end{align}

The dynamics are rigorously characterized in the Theorems \ref{oneThm} and \ref{mainThm} stated below.   The theorems together present a sharp polychotomy which delineates the stability of nine potentially ``strictly saturated'' distinct equilibria, along with a degenerate case where a continuum of equilibria exists, and the corresponding uniform persistence of variants in each case. First, we derive these nine potential equilibria with their corresponding component values.  In this way, we can capture all of the possible stable, non-negative equilibria and avoid listing equilibria that are always unstable or have negative components.   First, consider the equilibria types encountered in our analysis of the nested and strain-specific subnetworks adapted to this case of $n=2$ epitopes.   Note that any equilibrium, $\mathcal E^*$, with $z_2^*>0$ and $z_1^*=0$ or $y_3^*>0$ and $y_2^*=0$ will either be not saturated or have negative components when $s_1<s_2$.  Then there are four equilibria with one or less persistent immune response, ranging from the infection-free equilibrium $\bar{\mathcal E}_0$ to the ``escaped single immune'' equilibrium $\widetilde{\mathcal E}_1$.  Also, there are the nested and ``strain-specific'' type  equilibria $\overline{\mathcal E}_2, \widetilde{\mathcal E}_2$ and $\mathcal E^{\dagger}_2,\mathcal E^{\ddagger}_2$ (with 2 or 3 virus strains).  Next notice that Propositions \ref{prop33} and \ref{CycleProp} restrict the possibility of having an equilibrium with all four viral strains persistent to a degenerate case when $\mathcal R_1-\mathcal R_2-\mathcal R_3+\mathcal R_4=0$.   

We are left to check a new equilibrium type $\widehat{\mathcal E}_2$ where $y_1$, $y_2$, $y_3$ coexist, and by Proposition \ref{prop33}, its component $\widehat x$ can be derived as follows:
\begin{align}
C=\begin{pmatrix} A' & \vec{\mathcal R'}\end{pmatrix}^T=\begin{pmatrix} 1 & 0 & 1 \\ 1 & 1 & 0 \\ \mathcal R_1 & \mathcal R_2 & \mathcal R_3\end{pmatrix} \Rightarrow C^{-1}=\begin{pmatrix} * & * & -\mathcal R \\ * & * & \mathcal R \\ * & * & \frac{1}{\mathcal R}\end{pmatrix},  \ \ \widehat x=\vec 1^{\,T} C^{-1}_{(3)}=\frac{1}{\mathcal R},
\end{align}
where $\mathcal R= \mathcal R_2 + \mathcal R_3 - \mathcal R_1$.   The other components of $\widehat{\mathcal E}_2$ can be derived utilizing equilibrium equations \eqref{genEqcon}.  Therefore the distinct regimes for the 10 relevant equilibria forms are determined by the values of the viral and immune reproductive numbers, along with the following quantities:
\begin{align}
\mathcal Q_1&=1+\mathcal R_1s_1, \ \ \mathcal Q_2=\mathcal Q_1+ \mathcal R_2(s_2-s_1), \ \ \mathcal P_2 = 1 + s_1\mathcal R_3 + s_2 \mathcal R_2, \ \  \mathcal R= \mathcal R_2 + \mathcal R_3 - \mathcal R_1.   \label{2eqQ}
\end{align}
The following equilibria have two immune responses present:
 \begin{align}
  \widetilde{\mathcal E}_2  &= \left(\frac{1}{\mathcal R_3}, \widetilde{y}_1, \widetilde y_2, 0, \widetilde y_4, \widetilde{z}_1, \widetilde{z}_2 \right), \quad
\overline{\mathcal E}_2 = \left(\frac{1}{\mathcal Q_2},\ \bar y_1, \bar y_2, 0, 0, \  \bar{z}_1, \bar{z}_2 \right), \quad
 \widehat{\mathcal E}_2  = \left(\frac{1}{\mathcal R}, \ \widehat y_1, \widehat y_2, \widehat y_3, 0, \widehat z_1, \widehat z_2 \right) \label{2immEquilib} \\
\mathcal E^{\ddagger}_2  &= \left(\frac{1}{\mathcal R_3},0, y_2^{\ddagger}, y_3^{\ddagger}, y_4^{\ddagger}, z_1^{\ddagger}, z_2^{\ddagger}\right), \quad
 \mathcal E^{\dagger}_2  = \left(\frac{1}{\mathcal P_2}, \ 0, y_2^{\dagger}, y_3^{\dagger}, 0, z_1^{\dagger}, z_2^{\dagger}\right), \quad  \mathcal E(y_4^*)  = \left(\frac{1}{\mathcal R}, \  y_1^*, y_2^*,  y_3^*, y_4^*,  z_1^*,  z_2^* \right),  \notag \\
&\text{where}\quad \widetilde y_1=\bar y_1=s_1, \ \ \widetilde y_2=\bar y_2= s_2-s_1, \ \ \widetilde y_4=1- \frac{\mathcal Q_2}{\mathcal R_4}, \ \ \widetilde{z}_1= \frac{\mathcal R_1-\mathcal R_2}{\mathcal R_4},\ \ \widetilde{z}_2= \frac{\mathcal R_2}{\mathcal R_4}-1,\notag \\ & \bar{z}_1= \frac{\mathcal R_1-\mathcal R_2}{\mathcal Q_2},  \ \ \bar{z}_2= \frac{\mathcal R_2}{\mathcal Q_2}-1, \ \
 \widehat y_1= \frac{\mathcal P_2}{\mathcal R}-1, \ \ \widehat y_2=1-\frac{\mathcal Q_2}{\mathcal R}+s_2-s_1, \ \ \widehat y_3=1-\frac{\mathcal Q_2}{\mathcal R}, \notag \\ & \widehat z_1=  \frac{\mathcal R_1-\mathcal R_2}{\mathcal R}, \ \ \widehat z_2=  \frac{\mathcal R_1-\mathcal R_3}{\mathcal R}, \ \
  y_2^{\ddagger}=y_2^{\dagger}=s_2, \ \ y_3^{\ddagger}=y_3^{\dagger}=s_1, \ \ y_4^{\ddagger}=1-\frac{\mathcal P_2}{\mathcal R_4}, \ \ z_1^{\ddagger}=\frac{\mathcal R_3}{\mathcal R_4}-1, \ \ z_2^{\ddagger}=\frac{\mathcal R_2}{\mathcal R_4}-1, \notag \\ & z_1^{\dagger}=\frac{\mathcal R_3}{\mathcal P_2}-1, \ \ z_2^{\dagger}=\frac{\mathcal R_2}{\mathcal P_2}-1, \ \ y_1^*=\frac{\mathcal P_2}{\mathcal R}-1+y_4^*, \ \ y_2^*=1-\frac{\mathcal Q_2}{\mathcal R}+s_2-s_1-y_4^*, \ \ y_3^*=1-\frac{\mathcal Q_2}{\mathcal R}-y_4^*  \notag
 \end{align}
The four equilibria with one or zero immune responses are:
 \begin{align}
 \widetilde{\mathcal E}_1  &=\left(\frac{1}{\mathcal R_2}, s_1, 1-\frac{\mathcal Q_1}{\mathcal R_2}, 0, 0, \frac{\mathcal R_1-\mathcal R_2}{\mathcal R_2}, 0 \right) , \qquad \bar{\mathcal E}_1 = \left(\frac{1}{\mathcal Q_1}, s_1, 0, 0, 0, \frac{\mathcal R_1}{\mathcal Q_1}-1, 0 \right), \label{eqb7}\\
 \widetilde{\mathcal E}_0 &= \left(\frac{1}{\mathcal R_1}, 1-\frac{1}{\mathcal R_1}, 0, 0, 0, 0, 0 \right),  \qquad \bar{\mathcal E}_0  =\left(1, 0, 0, 0, 0, 0, 0 \right). \notag
  \end{align}

%\newpage
The theorems classifying dynamics for model (\ref{ode4}) with strict immunodominance, $s_1<s_2$, are as follows (proofs are in Appendix \ref{A2}):
\begin{theorem}[Stability of equilibria with one or zero immune response] \label{oneThm}
  Consider the model with two epitopes (\ref{ode4}) under the assumptions (\ref{conditions}) and suppose positive initial conditions, i.e. $x(0),y_i(0),z_j(0)>0$ for all $i= 1,\dots 4, j=1,2$.  Then the following results hold:
  \begin{itemize}
  \item[i.] If $\mathcal R_1\leq 1$, then $\bar{\mathcal E}_0$ is GAS (globally asymptotically stable).
  \item[ii.] If $1<\mathcal R_1 \leq \mathcal Q_1$, then $\widetilde{\mathcal E}_0$ is GAS.
  \item[iii.] If $\mathcal R_1 > \mathcal Q_1\geq \mathcal R_2$, then $\bar{\mathcal E}_1$ is GAS.
  \item[iv.] If $\mathcal Q_2 \geq \mathcal R_2> \mathcal Q_1$, then $\widetilde{\mathcal E}_1$ is GAS.
  \end{itemize}
  \end{theorem}

   \begin{theorem}[Stability of equilibria with two immune responses present]\label{mainThm}
 Consider the model with two epitopes (\ref{ode4}) under the assumptions (\ref{conditions}) and suppose positive initial conditions, i.e. $x(0),y_i(0),z_j(0)>0$ for all $i= 1,\dots 4, j=1,2$.  Then the stability of equilibria (\ref{2immEquilib}) are characterized as follows:
   \begin{enumerate}
 \item  if $\mathcal R < \mathcal Q_2$ and $\mathcal R_4 \leq \mathcal Q_2$, then $\bar{\mathcal E}_2$ is GAS.
  \item  if  $\mathcal R < \mathcal R_4$, then $\widetilde{\mathcal E}_2$ is GAS.

   \item  if $\mathcal Q_2 < \mathcal R < \mathcal P_2$ and  $\mathcal R_4 < \mathcal R$, then $\widehat{\mathcal E}_2$ is (locally) stable. Additionally, $\lim_{t\rightarrow\infty} x(t)=\widehat{x}=\frac{1}{\mathcal R}$, $\lim_{t\rightarrow\infty} y_4(t)=0$, and
   $$ \lim_{t\rightarrow\infty} \frac{1}{t} \int\limits_0^t y_i(s) \, ds = \widehat{y}_i, \ \ i=1,2,3, \ \  \lim_{t\rightarrow\infty} \frac{1}{t} \int\limits_0^t z_j(s) \, ds = \widehat{z}_j, \ \ j=1,2. $$
   Furthermore $y_i,z_j, \ \ i=1,2,3, j=1,2$ are uniformly persistent.
   \item  if $\mathcal P_2 < \mathcal R$ and $\mathcal R_4 \leq \mathcal P_2$, then $\mathcal E_2^{\dagger}$ is GAS.
  \item  if  $\mathcal Q_2 < \mathcal P_2 <\mathcal R_4 < \mathcal R$, then $\mathcal E_2^{\ddagger}$  is GAS.
  \item If $\mathcal R=\mathcal R_4$, then there is a continuum of saturated equilibria which forms a line connecting $y_4=0$ and $y_3=0$ boundaries at $\widehat y_3$ and $\widetilde y_4$, respectively, in the $(y_3, y_4)$ plane.
  \end{enumerate}
 \end{theorem}

%Bifurcation Fig was here

   \begin{table}[t!]
\centering
{\small
\begin{tabular}{ l|l|l|l|l|l|l|}
\toprule
Strains  &  $\mathcal R < \mathcal Q_2$, & $\mathcal R < \mathcal R_4$ & $\mathcal R = \mathcal R_4$ & $ \mathcal R_4<\mathcal R,$ & $\mathcal R_4 \leq \mathcal P_2,$ & $\mathcal Q_2 < \mathcal P_2$ \\ [0.5ex]
 &  $\mathcal R_4 \leq \mathcal Q_2 $ & & & $ \mathcal Q_2<\mathcal R<\mathcal P_2$ & $ \mathcal P_2 <\mathcal R$ & $<\mathcal R_4 <\mathcal R$ \\ [1.5ex]
\toprule
$x^*$ &   $\frac{1}{\mathcal Q_2}$ & $ \frac{1}{\mathcal R_4}$ & $ \frac{1}{\mathcal R}$& $ \frac{1}{\mathcal R}$&   $ \frac{1}{\mathcal P_2}$ & $ \frac{1}{\mathcal R_4}$ \\ [1.5ex]
$y_1^* \ \bf{(00)}$ &  $s_1$ & $s_1$ & $\frac{\mathcal P_2}{\mathcal R}-1+y_4^*$ & $\frac{\mathcal P_2}{\mathcal R}-1$ &  \bf{0}  &  \bf{0} \\ [1.5ex]
$y_2^* \ \bf{(10)}$ &  $s_2-s_1$ & $s_2-s_1$ & $1-\frac{\mathcal Q_2}{\mathcal R}+s_2-s_1-y_4^*$ & $1-\frac{\mathcal Q_2}{\mathcal R}+s_2-s_1$ & $s_2$ & $s_2$ \\ [1.5ex]	
$y_3^* \ \bf{(01)}$ &  \bf{0}   &    \bf{0} & $1-\frac{\mathcal Q_2}{\mathcal R}-y_4^*$ & $1-\frac{\mathcal Q_2}{\mathcal R}$& $s_1$& $s_1$ \\ [1.5ex]
$y_4^* \ \bf{(11)}$ &    \bf{0} & $1- \frac{\mathcal Q_2}{\mathcal R_4}$ & $y_4^*$ &  \bf{0} &   \bf{0}  &$1- \frac{\mathcal P_2}{\mathcal R_4}$ \\  [1.5ex] \hline \\ [-2ex]
Stable Equilib. &   $\bar{\mathcal E}_2$ & $\widetilde{\mathcal E}_2$ & $\mathcal E(y_4^*), \ 0\leq y_4^*\leq 1-\frac{\mathcal Q_2}{\mathcal R}$ &  $\widehat{\mathcal E}_2$&   $\mathcal E_2^{\dagger}$ & $\mathcal E_2^{\ddagger}$ \\ [0.5ex]
\bottomrule
\end{tabular}}
 \caption{Stable equilibria values for healthy and infected cells in Theorem \ref{mainThm} where six regimes sharply characterize distinct viral strain persistence scenarios ($\bf{0}$ components go extinct) when both immune responses $z_1$ \& $z_2$ persist (when $\mathcal R_2>\mathcal Q_2$).   Note that if $\mathcal R_2\leq Q_2$, then only $z_1$ and $y_1$, $y_2$ can persist and the dynamics are detailed in Theorem \ref{oneThm}. }
\label{table}
\end{table}

We make the following observations concerning the above theorems.  First, note that the inequalities in the hypotheses of Theorems \ref{oneThm} and \ref{mainThm} cover all possible parameter combinations under the conditions (\ref{conditions}).  Indeed, observe that  $\mathcal P_2 \leq \mathcal Q_2 \Leftrightarrow \mathcal R \leq 0$ and of course $\mathcal Q_2,\mathcal P_2>0$, therefore the case  $\mathcal P_2 \leq \mathcal Q_2$ falls under case 1 or case 2 of the theorem.  Additionally, $\mathcal R > \mathcal Q_2 \Rightarrow \mathcal R_2>\mathcal Q_2$, separating this case from the cases considered in Theorem \ref{oneThm}.   Note also that any equilibrium, $\mathcal E^*$, with $z_2^*>0$ and $z_1^*=0$ or $y_3^*>0$ and $y_2^*=0$ will not be saturated when $s_1<s_2$.

%\newpage

      \begin{figure}[t!]
      \subfigure[]{\label{fig10a}\includegraphics[width=7.5cm,height=4.2cm]{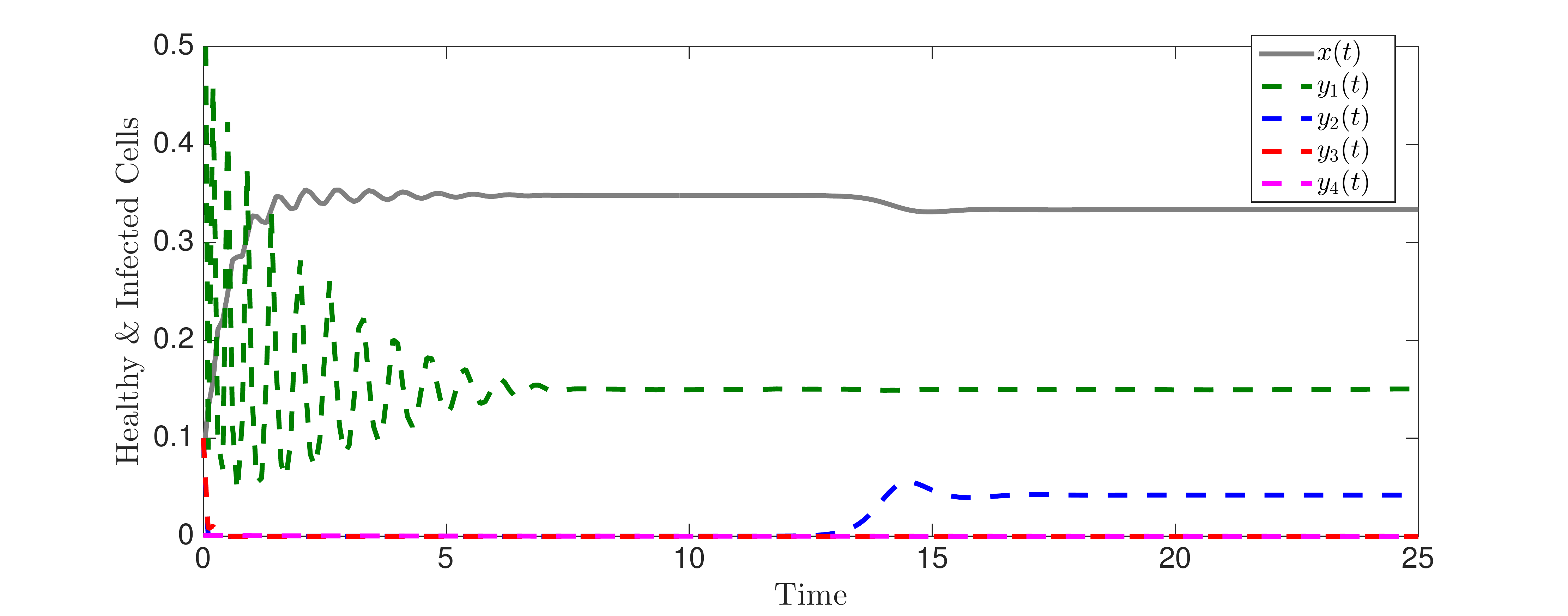}}
\subfigure[]{\label{fig10b}\includegraphics[width=7.5cm,height=4.2cm]{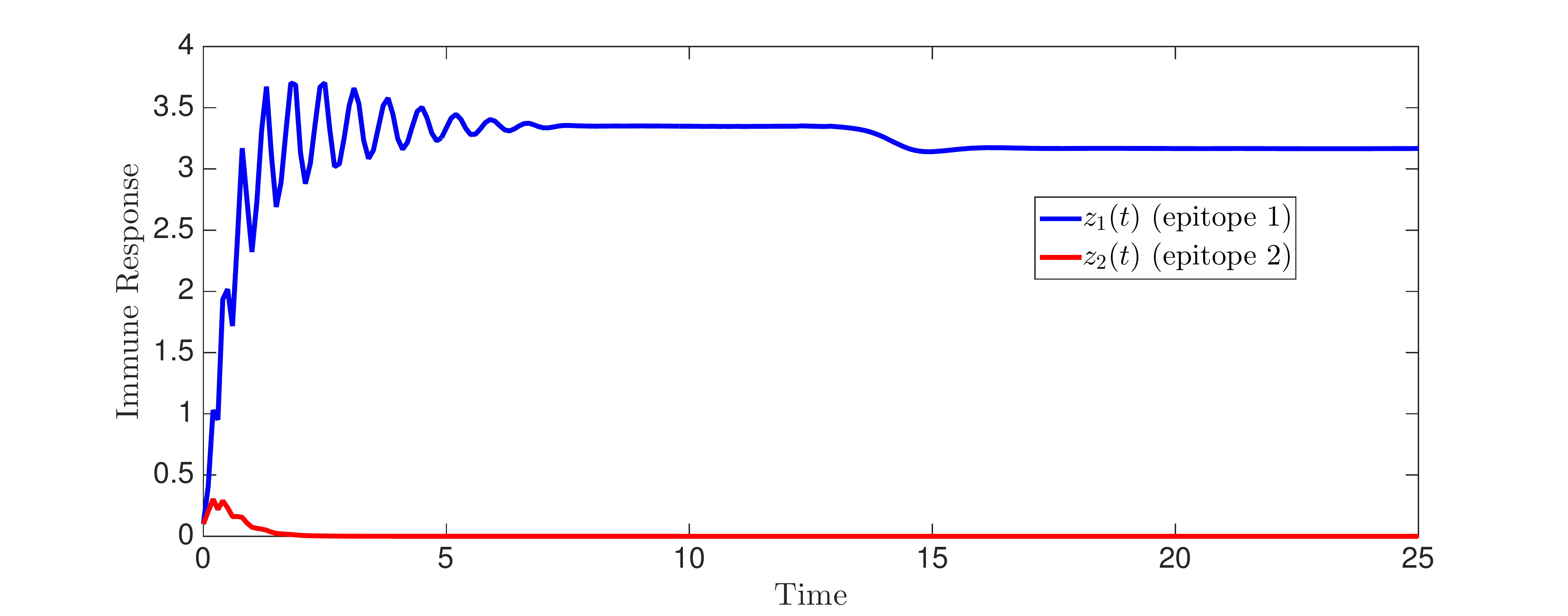}} \\
\subfigure[][]{\label{fig4a}\includegraphics[width=7.5cm,height=4.2cm]{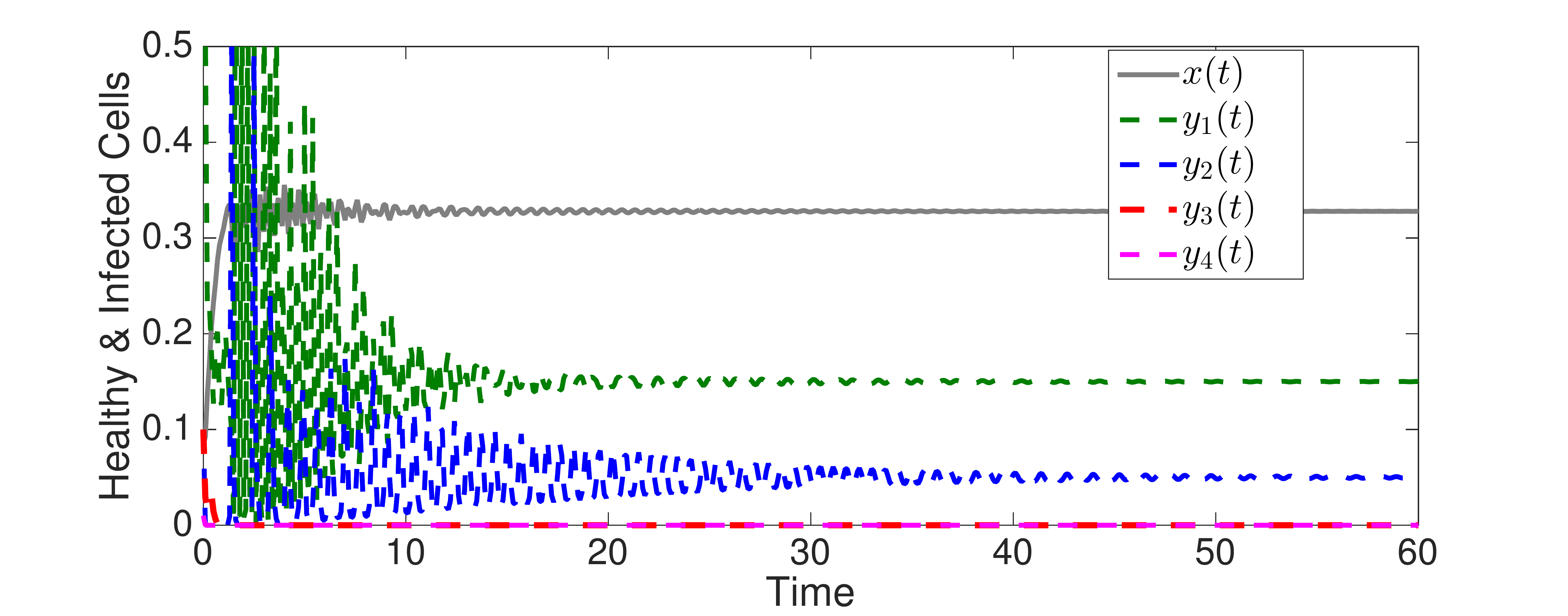}}
\subfigure[][]{\label{fig4b}\includegraphics[width=7.5cm,height=4.2cm]{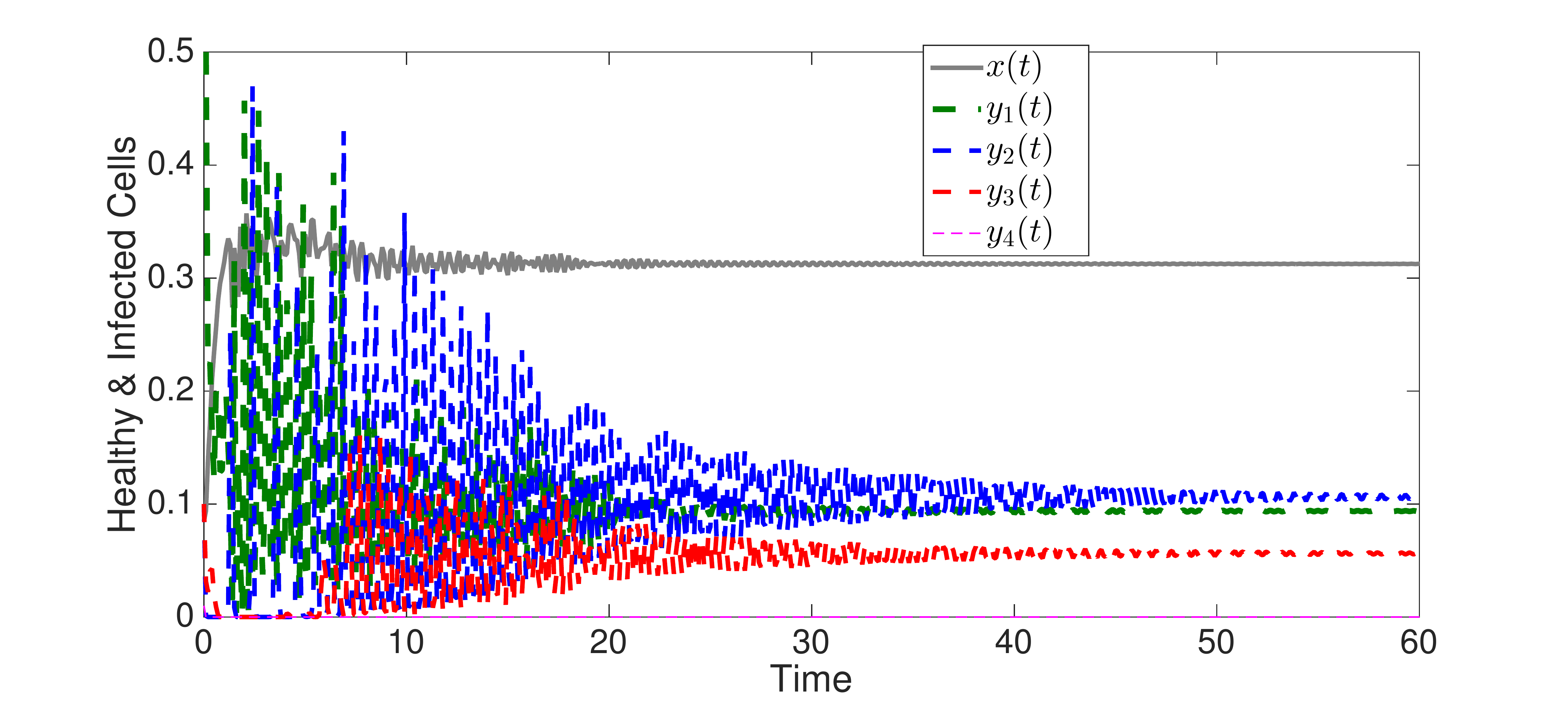}} \\
\subfigure[][]{\label{fig4d}\includegraphics[width=7.5cm,height=4.2cm]{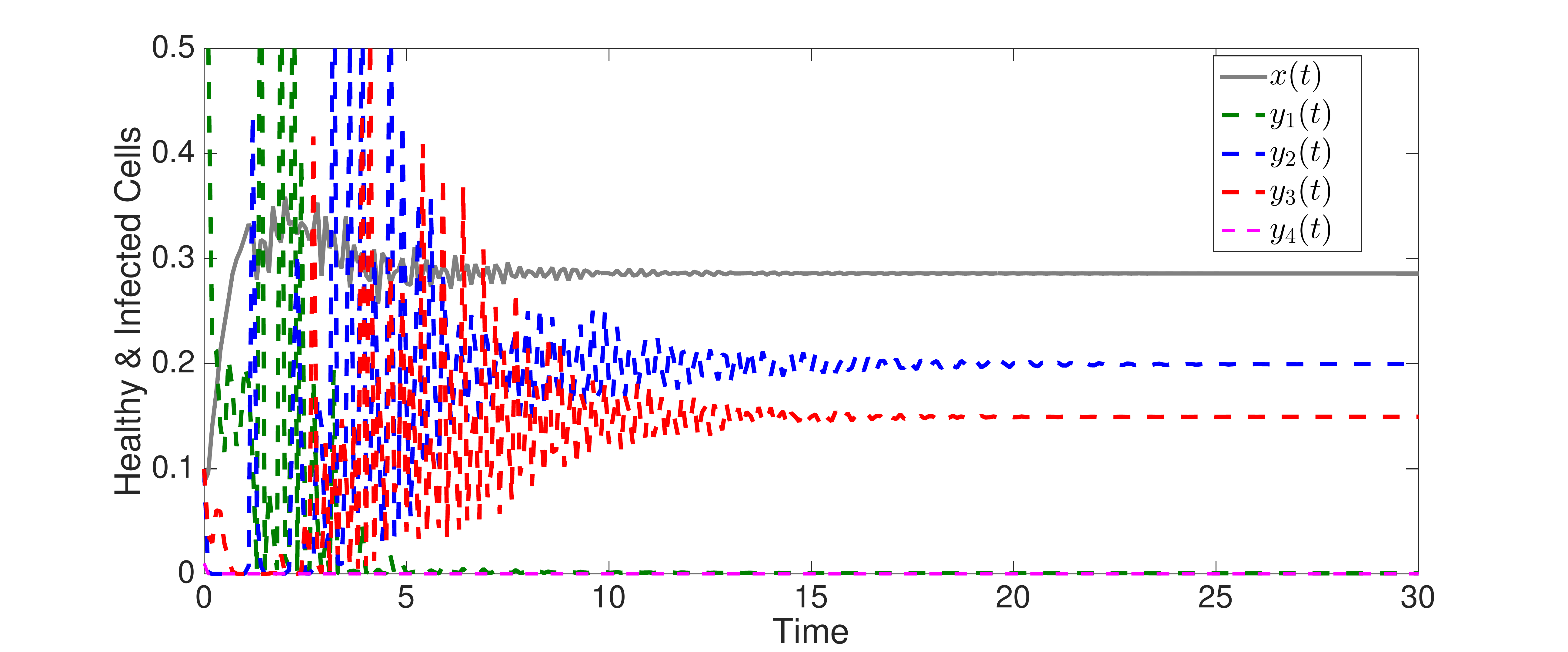}}
\subfigure[][]{\label{fig4d}\includegraphics[width=7.5cm,height=4.2cm]{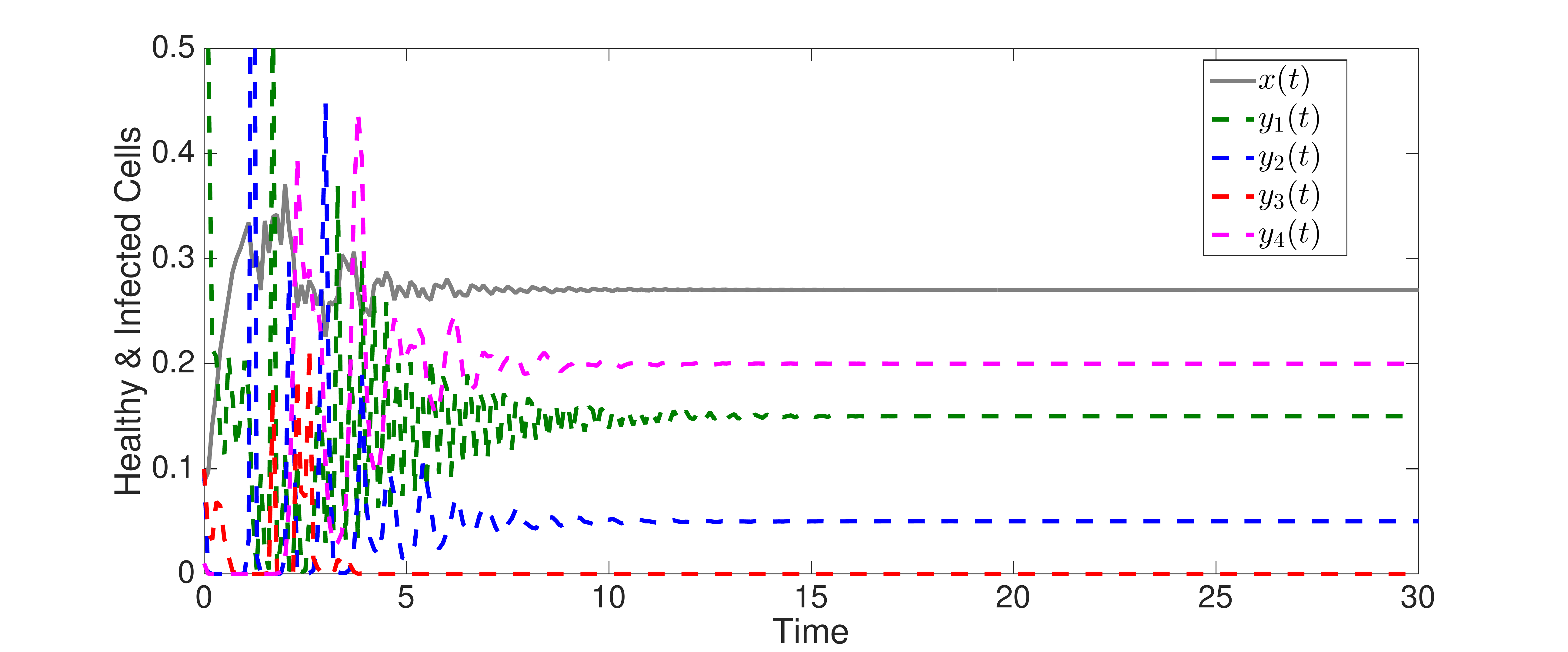}}
\caption{ (a) Healthy cells and Virus strains and (b) Immune response for simulations of system \eqref{ode4} depicting convergence to $\widetilde{\mathcal E}_1$.  Despite a ``single epitope advantage'' of $y_3$ for the corresponding parameter values, the mutant $y_2$ (resistant to immunodominant response $z_1$) excludes $y_3$ in the two-epitope scenario.  The parameters are as follows: $\mathcal R_1=12.5$, $\mathcal R_2=3$,  $\mathcal R_3=10$, $\mathcal R_4=2$, $\mathcal I_1=6.67$, $\mathcal I_2=5$, $\gamma_i=\gamma=70$ and $\sigma_j=10$ for $i=1,\dots 4$, $j=1,2$.  the initial conditions are $x(0)=\frac{1}{\mathcal R_1}$, $y_1(0)=1-\frac{1}{\mathcal R_1}$, $y_2(0)=y_3(0)=0.1$, $y_4(0)=0.001$, $z_1(0)=z_2(0)=0.01$ ($\approx \widetilde{\mathcal E}_0$, the wild-type viral steady state before immune response kicks in).  Here $\mathcal Q_1 < \mathcal R_2 <\mathcal Q_2$, and the comparison of \emph{single-epitope invasion eigenvalues} yields $\lambda_2=0.044\gamma < \lambda_3=1.86\gamma$, and single-epitope densities are $y_2^*=0.042 < y_3^*=0.65$.  Note that the delay in convergence to $\widetilde{\mathcal E}_1$ is due to $y_2$ transiently going to very low levels, and thus taking more time to invade.  (c) Convergence to $\overline{\mathcal E}_2$ after decreasing $\mathcal R_1$ from 12.5 to 12 and increasing $\mathcal R_2$ from 3 to 5  (d) Convergence to $\widehat{\mathcal E}_2$  after decreasing $\mathcal R_1$ from 12 to 11.8.  (d)  Convergence to $\mathcal E^{\dagger}_2$  after decreasing $\mathcal R_1$ from 11.8 to 11.5. (f) Convergence to $\widetilde{\mathcal E}_2$  after increasing $\mathcal R_4$ to 3.6.  Note $z_1$ and $z_2$ both persist for (c)-(f).}
\label{fig4}
  \end{figure}

  For the case $\mathcal R_4=\mathcal R$ and $\mathcal Q_2 < \mathcal R < \mathcal P_2$ , both equilibria $\widetilde{\mathcal E}_2$ and $\widehat{\mathcal E}_2$ are saturated, but the inequalities (\ref{inequ}) are not strict.  In this case, there are a continuum of saturated (neutrally stable) non-negative equilibria with $x(t)\rightarrow\frac{1}{\mathcal R}$ and $y_4(t)\rightarrow 1- \frac{\mathcal Q_2}{\mathcal R} - y_3^*$, where $0\leq y_3^* \leq 1- \frac{\mathcal Q_2}{\mathcal R}$.  The endpoints of the line of equilibria in the $(y_3,y_4)$ plane correspond to $\widetilde{\mathcal E}_2$ and $\widehat{\mathcal E}_2$, where $\widetilde y_3=0$ and $\widehat y_4=0$, respectively.  The bifurcation as $\mathcal R_4$ increases above $\mathcal R$, where Case 3 transitions to Case 6 and then to Case 2 of Theorem \ref{mainThm}, is illustrated in Figure \ref{fig3} in Appendix \ref{A2}.  When $\mathcal R_4\neq \mathcal R$, it is not possible for all $y_i$ to persist, as the dynamics will fall into one of the other cases in Theorems \ref{oneThm} and \ref{mainThm}.  These results illustrate Proposition \ref{CycleProp}, which precludes the possibility of equilibria with all components $y_{i}$ positive except in the case $\mathcal R_4\neq \mathcal R$ since the viral strains form a 4-cycle in the associated hypercube graph.

  An \emph{interesting finding} is that strain \emph{$y_2 \ (10)$  is present in all equilibria (and uniformly persistent) with mutation}, even if it has a much higher fitness cost than $y_3 \ (01)$, i.e. $\mathcal R_3>>\mathcal R_2$.   Thus, if escape occurs in the two epitope setting, the viral quasispecies always includes the mutant $y_2$ with resistance to the immunodominant epitope, $z_1$.  We can characterize the (linearized) invasion rate of an escape mutant, say $y_2$, from the single epitope, $z_1$, by calculating $\frac{\dot y_2}{y_2}$ at equilibrium $\overline{\mathcal E}_1$, to obtain $\lambda_2=\gamma_2\left(\frac{\mathcal R_2}{\mathcal Q_1}-1\right)=\frac{\mathcal R_2}{1+\mathcal R_1s_1}-1$. The mutant $y_2$ then converges to $y_2^*=1-\frac{1+\mathcal R_1s_1}{\mathcal R_2}$, in equilibrium $\widetilde{\mathcal E}_1$.  The comparable quantities for mutant $y_3$ in escaping the single epitope $z_2$ (in the absence of $z_1$) are $\lambda_3=\gamma_3\left(\frac{\mathcal R_3}{1+\mathcal R_1s_2}-1\right)$ and $y_3^*=1-\frac{1+\mathcal R_1s_2}{\mathcal R_3}$.  Even when $y_3$ (escaping $z_2$) has larger invasion characteristics in the single epitope escape setting ($y_3^*>y_2^*$ and $\lambda_3>\lambda_2$), $y_2$ will still persist and may actually exclude $y_3$ in the two epitope model.  Figure \ref{fig10a},\ref{fig10b} depict simulations of this scenario of $y_2$ excluding $y_3$ despite the larger selective advantage of $y_3$ in the single epitope setting.   Figure \ref{fig4} also shows further simulations of (\ref{ode4}), which are all consistent with assertions of Theorem \ref{oneThm} and Theorem \ref{mainThm}.

  There are two special cases of model (\ref{ode4}) where analysis further suggests the superior role of immunodominance over viral reproductive fitness.  First, if we relax the assumption of strict immunodominance hierarchy, i.e. we allow that $s_1=s_2$, then a ``strictly saturated'' non-negative equilibrium $\mathcal E^*$ has the following properties: $z_1^*>0\Leftrightarrow z_2^*>0$ and $y_2^*>0\Leftrightarrow y_3^*>0$.  The calculations for this case are presented in the Appendix \ref{A3}.  Essentially, non-trivial strictly saturated equilibria are of the form $\widehat{\mathcal E}_2$,   $\mathcal E_2^{\dagger}$ or $\mathcal E_2^{\ddagger}$ with the corresponding parameter regimes as defined in Table \ref{table}.  Therefore, in the case of equal immunodominance ($s_1=s_2$), no matter the fitness  of the distinct mutant strains $y_2$ and $y_3$, both strains can only persist together.  Next, consider the case where the viral fitness cost to each epitope is equal, say the cost is $c=1-f$, and the strict immunodominance holds, $s_1>s_2$.  Then $\mathcal R_2=\mathcal R_3=f\mathcal R_1$ and $\mathcal R_4=f^2\mathcal R_1$.  This scenario will be analyzed in further generality for $n$ epitopes in the next section.  The result for this special case of model (\ref{ode4}) is that one of the nested equilibria, $\overline{\mathcal E}_i$ or $\widetilde{\mathcal E}_i$ ($i=0,1,2$) will be globally asymptotically stable.   Thus, in contrast to the previous special case ($s_1=s_2$) where viral fitness did not determine persistence, here (in this special case of equal viral fitness costs) we find that persistence is determined by immunodominance.

     The above results indicate that the \emph{immunodominance hierarchy is the most important factor in directing epitope escape, more so than viral fitness cost}.  In addition, the persistent variants depend upon reproductive numbers and quantities defined in (\ref{2eqQ}), implying that the escape pathway depends upon the entire interacting system, not just the parameters associated with the single epitopes and corresponding resistant viral mutant strains.  Both the relative importance of immunodominance and the multi-epitope dependence align with findings in an \emph{in vivo} study of HIV patients \cite{liu2013vertical}.  From a broader ecological point of view, our results suggest top-down control of food webs, where top predators have more influence than intermediate species, along with the interconnectedness of complex ecological networks.

%\FloatBarrier

%\newpage

\subsection{Uniform fitness costs for escape on full $n$-epitope network}\label{EqFitSec}

We consider the full network for $n$ epitopes, in the special case where each viral mutation of an epitope incurs a uniform fitness cost, $c=1-f$, where $f\in(0,1)$ is the ratio between reproduction number of mutant and descendent strain.  In particular, indexing the wild-strain here as $y_1=y_w$ with fitness $\mathcal R_1$, then we make the following assumption on the $2^n$ viral strains in the full network:
\begin{align}
d(y_i,y_1)=p \Rightarrow \mathcal R_i=f^p\mathcal R_1, \quad \forall i\in[1,2^n], \label{FitCost}
\end{align}
  where $d(y_i,y_1)$ is the Hamming distance between the associated binary sequences of the viral strains as defined earlier.  Note that since the wild-strain $y_1$ is susceptible at all epitopes ($y_1\sim 0\cdots 0$), a viral strain $y_i$ with $d(y_i,y_1)=p, \ p\in[0,n]$ has mutated $p$ epitopes and thus has a (susceptible) epitope set of cardinality $n-p$, i.e. $|\Lambda_i|=n-p$.  Then, we can write the model (\ref{odeS}) as follows:
  \begin{align}
\dot x &= 1-x- x\sum_{i=1}^{2^n} \mathcal R_i y_i,  \label{odeU} \\
 \dot y_i& = \gamma_i y_i\left(f^p\mathcal R_1 x -1 -\sum_{j=1}^{n-p}  z_{i_j} \right), \quad \ i=1,\dots,2^n, \ \ d(y_i,y_1)=p,  \ \ \Lambda_i=\left\{i_1,\dots,i_{n-p}\right\}, \ \  p\in[0,n],  \notag \\
  \dot z_j &= \frac{\sigma_j}{s_j} z_j\left(\sum_{i: j\in \Lambda_i} y_i -s_j \right), \quad j=1,\dots,n. \notag
\end{align}

  As before, we assume the immunodominance hierarcy (\ref{immunodom}), $s_1\leq\cdots\leq s_n$.  We denote the viral strains associated with ``nested'' (sequential) escape, in addition to the wild strain $y_1$, as $y_2,\dots,y_{n+1}$.  This means that for $i=1,\dots,n$, the epitope set of $y_i$ is $\Lambda_i=\left\{ i,\dots, n \right\}$ (having escaped immune responses $z_1,\dots,z_{i-1}$) and $\Lambda_{n+1}=\emptyset$ {\tt so $y_2\thicksim (1,0,\cdots,0), y_3=(1,1,0,\cdots,0),\cdots, y_{n+1}=(1,1,\cdots,1)$. }

  Our main result of this section is, informally, that viral escape from $n$ epitopes follows a ``nested pattern'' when the immunodominance hierarchy (\ref{immunodom}) is strict. In other words, one of the nested equilibria, $\widetilde{\mathcal E}_{k+1}$ or $\overline{\mathcal E}_{k}$ (introduced in Section \ref{NestedSec}) is stable. Roughly speaking, if viral fitness costs of resistance are the same independent of CTL type, it is better for a mutant strain to escape the \emph{current} dominant
  CTL type than a sub-dominant one. 
   In the Appendix \ref{A4}, we show that two other classes of equilibria, namely ``strain-specific'' and ``one-mutation'' equilibria, are unstable in this case of equal fitness costs, even when inequalities (\ref{immunodom}) are not strict.  Now the main theorem is stated below.

\begin{theorem}\label{ThmEqFit}
Consider model (\ref{odeU}), the full network on $n$ epitopes ($m=2^n$) with equal fitness costs (\ref{FitCost}) and strict immunodominance hierarchy (strict inequalities in (\ref{immunodom})).  Suppose $y_i$, $i\in [1,n+1]$, is indexed so that $\Lambda_i=\left\{ i,\dots, n \right\}$ for $i=1,\dots,n$, $\Lambda_{n+1}=\emptyset$.  Then  $y_i(t)\rightarrow 0$ as $t\rightarrow \infty$ for all $i\in[n+2,2^n]$, and Theorem \ref{NestedThm} holds in system (\ref{odeU}).  In particular, $y_i,z_i$ are uniformly persistent for $1\leq i \leq k\leq n$ when $\mathcal R_k>\mathcal Q_k$ (and $y_{k+1}$ is also persistent if $\mathcal R_{k+1}>\mathcal Q_k$).
\end{theorem}

\begin{proof}
 If $\mathcal R_1> \mathcal Q_1$, let $k$ be the largest integer in $[1,n]$ such that $f^{k-1}\mathcal R_1=\mathcal R_k > \mathcal Q_k$, otherwise let $k=0$.  We will apply Theorem \ref{genThm}.  It suffices to check the invasion rate for $y_{\ell}$ in (\ref{inequ}) for $\ell\in [n+2,2^n]$ since Theorem \ref{NestedThm} establishes the result in the perfectly nested submodel consisting of $y_1,\dots,y_{n+1}$ and $z_1,\dots,z_n$.  First suppose that $f^k\mathcal R_1=\mathcal R_{k+1}\leq \mathcal Q_k$.  Let $\ell\in [n+2,2^n]$ and it suffices to consider $\Lambda_{\ell} \cap [1,k]$ since the calculations will be considered at equilibrium $\overline{\mathcal E}_k$ where $z_i^*=0$ for all $i\geq k+1$.  Suppose that $d(y_{\ell},y_1)=p$ where $p\in [1,n-1]$.  Then $\mathcal R_{\ell}=f^p \mathcal R_1$.   Consider the invasion rate for $y_{\ell}$ at the equilibrium $\overline{\mathcal E}_k$:
\begin{align}
\frac{\dot y_{\ell}}{\gamma_{\ell} y_{\ell}}&= \frac{f^p \mathcal R_1}{\mathcal Q_k} -1 - \sum_{i\in\Lambda_{\ell}\cap [1,k]} \overline z_i \notag
\end{align}
If $p>k$, then clearly $$\frac{\dot y_{\ell}}{\gamma_{\ell} y_{\ell}}\leq \frac{f^p \mathcal R_1}{\mathcal Q_k} -1< \frac{f^k \mathcal R_1}{\mathcal Q_k} -1\leq 0,
$$ {\tt by the definition of $k$.}
If $p\leq k$, $|\Lambda_{\ell}|=n-p\geq n-k$ and note that $\Lambda_{\ell}\neq \Lambda_{k+1}=\left\{ k+1,\dots, n \right\}$ (where $\Lambda_{k+1}$ is the epitope set of $y_{k+1}$).  Thus,  $[1,k]\cap\Lambda_{\ell}\neq \emptyset$.  By (\ref{equilib2}), for $2\leq i\leq k-1$,
\begin{align}
\overline z_i - \overline z_{i-1}&= \frac{\mathcal R_{i}-\mathcal R_{i+1}-(\mathcal R_{i-1}-\mathcal R_{i})}{\mathcal Q_k} \notag \\
&=  \frac{-f^{i-2}\mathcal R_1}{\mathcal Q_k}\left(f^2-2f+1\right) \notag \\
&=  \frac{-f^{i-2}\mathcal R_1}{\mathcal Q_k}\left(f-1\right)^2 < 0 \quad \text{for} \ \ f\in(0,1). \notag
\end{align}
 Similarly $\overline z_k-\overline z_{k-1}<0$, and thus at equilibrium $\overline{\mathcal E}_k$, we find $$\overline z_k<\overline z_{k-1}<\dots < \overline z_{1}.$$  If $p=k$, then
\begin{align}
\frac{\dot y_{\ell}}{\gamma_{\ell} y_{\ell}}&\leq \frac{f^p \mathcal R_1}{\mathcal Q_k} -1 -  \overline z_k  \notag\\
& \leq \frac{f^{k} \mathcal R_1}{\mathcal Q_k} -1 - \left(\frac{f^{k-1}\mathcal R_1}{\mathcal Q_k}-1\right)   \notag \\
& <0 \quad  \text{for} \ \ f\in(0,1). \notag
\end{align}
 If $p\leq k-1$, then {\tt $[1,k-1]\cap\Lambda_{\ell}\neq \emptyset$ so}
\begin{align}
\frac{\dot y_{\ell}}{\gamma_{\ell} y_{\ell}}&< \frac{f^p \mathcal R_1}{\mathcal Q_k} -1 -  \overline z_k  \notag\\
& \leq \frac{f^{k-1} \mathcal R_1}{\mathcal Q_k} -1 - \left(\frac{f^{k-1}\mathcal R_1}{\mathcal Q_k}-1\right)   \notag \\
& =0 \quad  \text{for} \ \ f\in(0,1). \notag
\end{align}
Therefore, the equilibrium $\overline{\mathcal E}_k$ is saturated with inequalities (\ref{inequ}) strictly holding and $\overline{\mathcal E}_k$ is unique in its positivity class.  Thus Theorem \ref{genThm} can be applied to obtained the conclusions of Theorem \ref{NestedThm}.  If $f^k\mathcal R_1=\mathcal R_{k+1}> \mathcal Q_k$, then a similar argument works with $x^*=\frac{1}{\mathcal R_{k+1}}$ instead of $x^*=\frac{1}{\mathcal Q_k}$, which shows that $\widetilde{\mathcal E}_{k+1}$ is stable in that case.
 \end{proof}

In Figure \ref{fig6}, we illustrate Theorem \ref{ThmEqFit} by numerical solution of the model (\ref{odeU}) in the case of $n=3$ epitopes.  The simulations show that after some transient dynamics, only the viral strains associated with the nested network, $y_1,y_2,y_3,y_4$, persist.  In other words, the full network of $n=3$ epitopes (displayed in Figure \ref{fig2a}) converges to the perfectly nested subnetwork (displayed in Figure \ref{fig2c}).

%Put figure with diagram and simulations, discuss implications of Thm.
   \begin{figure}[t!]
\subfigure[][]{\label{fig6a} \includegraphics[width=7.5cm,height=4cm]{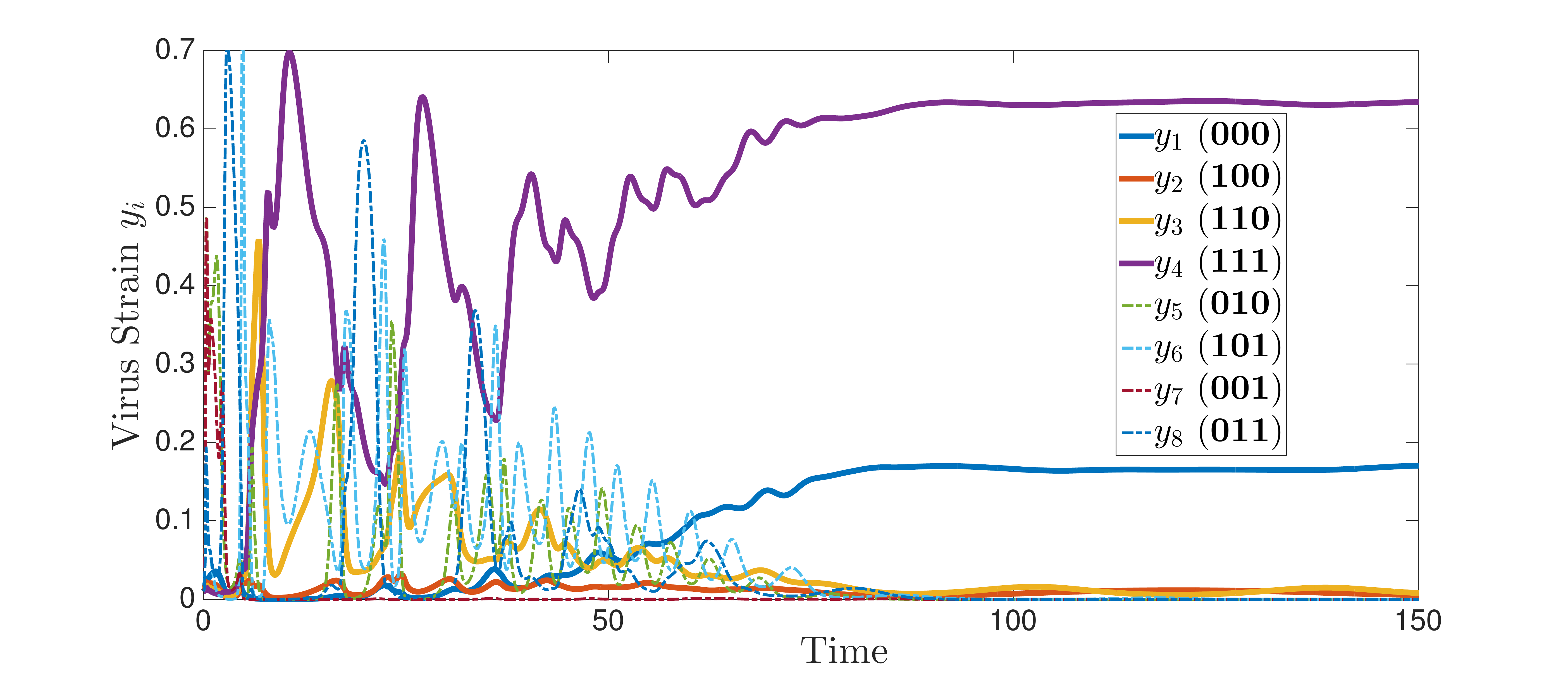}}
\subfigure[][]{\label{fig6b} \includegraphics[width=7.5cm,height=4cm]{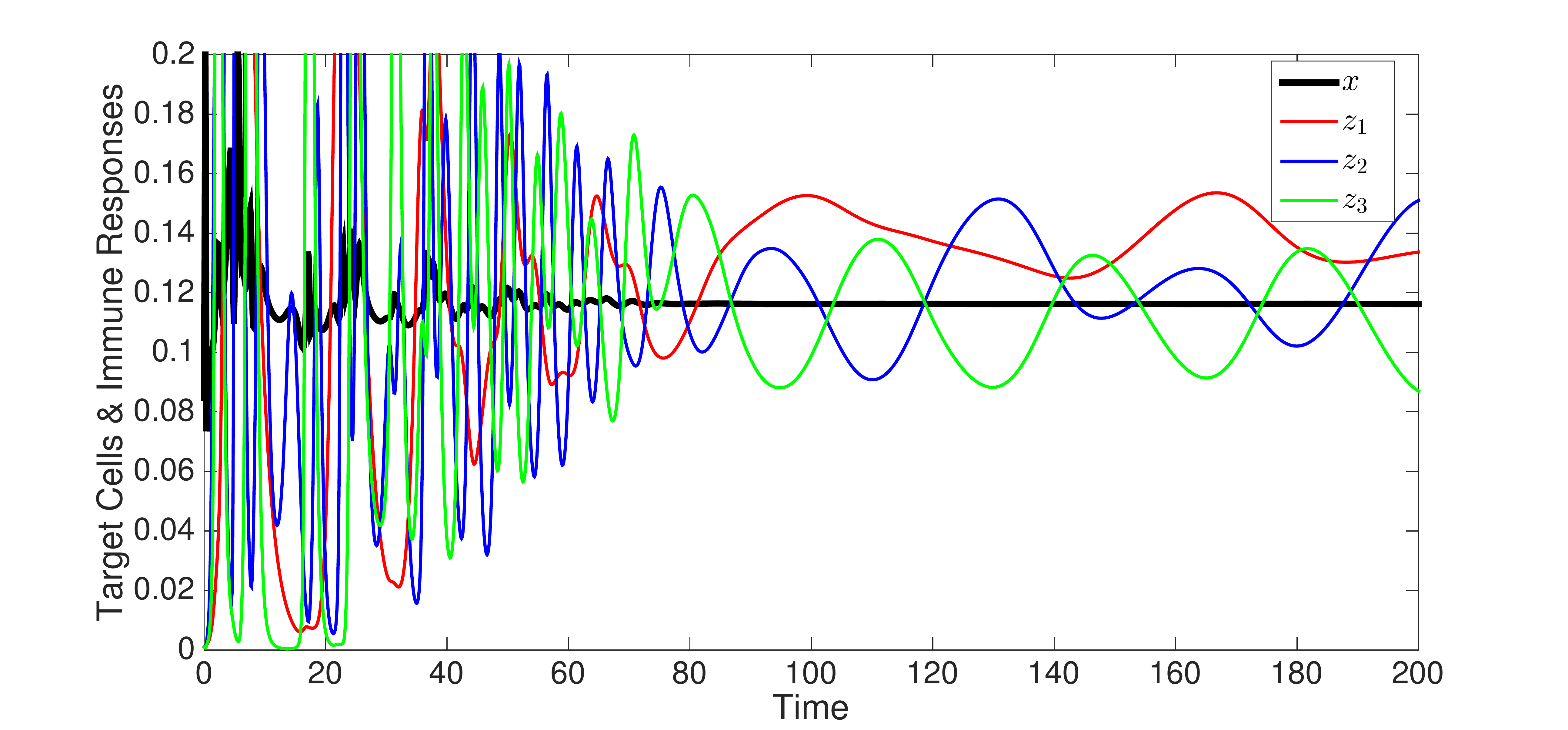}}
%\subfigure[][]{\label{fig6c} \includegraphics[height=3.6cm,width=5.3cm]{eqfitcost_z_long.eps}}
\caption{\footnotesize{Example dynamics in case of uniform fitness cost and $n=3$ epitopes.  The solutions converge to the ``nested equilibrium'' $\widetilde{\mathcal E}_4$ where $y_1,y_2,y_3,y_4$ persist and other $y_i$ go extinct.  Note that there are transient oscillations in the persistent viral and immune variant populations for a large period of time before convergence to the equilibrium.  The parameters are as follows: wild-type viral fitness $\mathcal R_1=11.8$, mutant fitness proportion $f=0.9$, $s_1=0.167$, $s_2=0.175$, $s_3=0.185$, $\gamma_i=3.5$ for $1\leq i\leq 4$, $\gamma_i=17.5$ for $5\leq i\leq 8$, $\sigma_1=0.5$, $\sigma_2=\sigma_3=2.5$.  Initial conditions are $x(0)=1$, $y_i(0)=0.01$, $z_j(0)=0.001$, $1\leq i\leq 8$, $1\leq j\leq 3$. }}
  \label{fig6}
  \end{figure}

  \section{Discussion}
  In this paper, we analyzed a virus model consisting of target cells, multiple virus strains and several immune response populations.  The interaction of virus and immune response is described by a network reflecting the avidity of each distinct immune response in recognizing each particular virus strain.  We find some general conditions on stability and feasibility of equilibria, along with uniformly persistent virus and immune response variants by utilizing Lyapunov function techniques.

 We specialize the model to consider the scenario where the immune response populations are $n$ different CTL lines, each specific to a particular epitope, and there are $2^n$ virus strains containing all possible combinations of (resistance conferring) epitope mutations.  In this case, the viral strains in the virus-immune network can be translated to an $n$-dimensional hypercube graph representative of the potential pathways of immune escape by the virus.  The number of uniformly persistent viral strains and CTL populations can be built up in an ordered fashion dependent on derived invasion thresholds for ``strain-specific'' and ``perfectly nested'' subgraphs.   For the full network including $2^n$ viral strains, in certain cases we sharply characterize dynamics and stability of distinct equilibria.  In particular, we find that for 
 \begin{itemize}
 \item[(i)] \emph{$n=2$ epitopes:}  the escape pathway (persistent viral strain set) always includes the mutant resistant to the immunodominant epitope, even when it suffers a relatively high fitness cost,
 \item[(ii)] \emph{equal fitness costs for each of $n$ epitopes:}  the network of $2^n$ viral strains always converges to a perfectly nested subgraph (ordered according to immunodominance) with less than or equal to $n+1$ strains.
 \end{itemize}

 Our theoretical results indicate that a diverse viral ``quasispecies'' can be built through resistance mutations at multiple epitopes and the immunodominance hierarchy is the most important factor determining the escape pathway.  Prior research on HIV data has found that CTL immune responses drive within-host HIV diversity \cite{pandit2014reliable} and immunodominance hierarchy is the major factor shaping viral escape from multiple epitopes \cite{Barton,liu2013vertical}.  Thus our model framework and analysis can shed light on the dynamic network of interacting virus and immune response variants, which may be important to understand for development of effective immunotherapies.

 Future research can build upon the results presented here in several ways.  First, since rapid evolution of HIV due to CTL pressure is motivation for our model, mutation between the different viral strains can be explicitly included in the model.  Preliminary simulations from stochastic versions of the model show that qualitative dynamics are preserved under mutation.  Rigorous global perturbation arguments in the deterministic setting may be attempted to show the effect of small mutation rates, however for the general case, this would rely upon the conjecture of global stability for equilibria.  Thus, another important theoretical question is proving global stability when uniform persistence occurs.  Finally, further extensions of our work can be applicable to coevolution of virus and cross-reactive antibodies during HIV infection \cite{luo2015competitive}, along with other complex ecological networks.

\section*{Acknowledgement}
We would like to thank the editor and anonymous reviewers for comments which have led to an improved manuscript.  We also acknowledge support by the  Simons Foundation Grant 355819.

\appendix
\section{Appendix}
\subsection{Existence of saturated equilibrium} \label{A1}
\begin{proof}[Proof of Proposition \ref{SatExist}]
We find it somewhat easier to work with the unscaled system (\ref{ode2}) so we introduce
some new parameters into it. Let $\epsilon>0$ be so small that $b-\epsilon(m+n)>0$
and $0\le \lambda\le 1$ be a homotopy parameter. Our perturbed system is given by:
\begin{eqnarray*}
% \nonumber to remove numbering (before each equation)
  X' &=& b-cX -\lambda X \sum_i \beta_iY_i-\epsilon(m+n) \\
  Y_i' &=& \lambda Y_i\left(\beta_i X-\sum_j r_{ij}Z_j \right)-\delta_i Y_i+\epsilon , 1\le i\le m,\\
  Z_j' &=& \lambda q_jZ_j \sum_i r_{ij}Y_i-\mu_j Z_j+\epsilon q_j , 1\le j\le n.
\end{eqnarray*}
We refer to the vector field on the right side as $G(W,\lambda,\epsilon)$, where $W=(X,Y,Z)\in \mathbb{R}^{m+n+1}_+$. Then
$G(W,1,0)$ is the vector field given in equations (\ref{ode2}).

Straightforward calculation establishes that
\begin{eqnarray*}
% \nonumber to remove numbering (before each equation)
  (X+\sum_i Y_i+\sum_j Z_j/q_j)' &=& b-cX-\sum_i \delta_iY_i-\sum_j\frac{\mu_j}{q_j}Z_j \\
   &\le& b-d(X+\sum_i Y_i+\sum_j Z_j/q_j)
\end{eqnarray*}
for $d=\min\{c,\delta_i,\mu_j\}$. Fix $p>b/d$ and let
$$
\mathcal U=\{W\in \mathbb{R}^{m+n+1}_+: X+\sum_i Y_i+\sum_j Z_j/q_j\le p\}
$$

We claim that there are no equilibria of $G(W,\lambda,\epsilon)$ on the boundary of $\mathcal U$.
Notice that any equilibrium $E=(X,Y,Z)\in \mathcal U$ satisfies $Y_i>0$ and $Z_j>0$ because of the perturbation terms and because $b > \epsilon(m + n)$. If $E$ belongs to the boundary
of $\mathcal U$, then either $X=0$ or $X+\sum_i Y_i+\sum_j Z_j/q_j=p$. Each of these is easy to rule out.
Suppose, for example the latter holds. Then the left side of the differential inequality above vanishes
at $E$ so we have $0\le b-pb$, a contradiction to $p>b/d$.

Now we can employ degree theory, as in Hofbauer and Sigmund \cite{hofbauer1998evolutionary}. By Homotopy invariance of degree
we have that $\hbox{deg}(G(\bullet,1,\epsilon),\mathcal U)= \hbox{deg}(G(\bullet,0,\epsilon),\mathcal U)$
and the latter is easy to compute since it is linear and has a unique equilibrium $\tilde E\in \mathcal U$
$$
\hbox{deg}(G(\bullet,0,\epsilon),\mathcal U)=\hbox{sgn} \det G_W(\tilde E,0,\epsilon)=\hbox{sgn} [(-1)^{m+n+1}c\prod_i \delta_i\prod_j \mu_j]=(-1)^{m+n+1}
$$

As the degree is nonzero, $G(\bullet,1,\epsilon)$ has at least one equilibrium in the interior of $\mathcal U$ for each
small $\epsilon$. Now, the argument for the existence of a saturated equilibrium of $G(\bullet,1,0)$
follows as in Theorem 13.4.1 in Hofbauer \& Sigmund's text \cite{hofbauer1998evolutionary} by taking limits as $\epsilon\to 0$.
\end{proof}

\subsection{Strain-specific subnetwork}\label{ssA}
\begin{proof}[Proof of Proposition \ref{ssF}]
Consider system \eqref{odeS} under the prescribed assumptions of Proposition \ref{ssF}.  We claim that any equilibrium $\mathcal E^*$ with $\Omega_z\subseteq [1,n]$ and $\Omega_y\subseteq [1,n+1]$ is saturated only if $\Omega_z= [1,n]$ and $\Omega_y\supseteq [1,n]$.  In other words $\mathcal E^{\dagger}_n$ and  $\mathcal E^{\ddagger}_{n+1}$ are the only strain-specific equilibria which can be saturated in the full hypercube network.  Suppose by way of contradiction that we can choose $i,j\in [1,n]$ such that $i\in \Omega_z$ and $j\notin\Omega_z$.  Then consider strain $y_{\ell}$ such that $\Lambda_{\ell}=\left\{i,j\right\}$, i.e. $y_{\ell}\leftrightarrow y_i$ and $y_{\ell}\leftrightarrow y_j$.  At equilibrium $\mathcal E^*$, its invasion rate (on of the ``saturated'' conditions \eqref{inequ}) is given by $\frac{\dot y_{\ell}}{\gamma_{\ell} y_{\ell}}= (\mathcal R_{\ell}-\mathcal R_i)/\mathcal P$, where $\mathcal P= \mathcal P_{\mathcal J}$ or $\mathcal R_{n+1}$ with $\mathcal J\subset [1,n]$.  Since $\mathcal R_{\ell}>\mathcal R_i$, equilibrium $\mathcal E^*$ can not be saturated.  The conditions for equilibria $\mathcal E^{\dagger}_n$  to be saturated is as follows in the general model:
\begin{align}
 \left( |\Lambda_{\ell}| - 1 \right) \mathcal P_n + \mathcal R_{\ell} \leq  \sum_{i\in\Lambda_{\ell}} \mathcal R_i  \quad \forall \ell \in [n+1, 2^n]. \notag
 \end{align}
 An analogous statement holds for $\mathcal E^{\ddagger}_{n+1}$ in case (ii).  
 \end{proof}

\subsection{Two-epitope model: dynamics when $s_1<s_2$} \label{A2}
\begin{proof}[Proof of Theorems \ref{oneThm} and \ref{mainThm}]
We apply Theorem \ref{genThm} for each equilibrium and case. Note that the Lyapunov function $W$ is the same as in the proof of Theorem \ref{genThm} except that
$s_j$ replaces $\rho_j$.  First in Theorem \ref{oneThm}:

\textbf{Case i.}:
\begin{align*}
\dot W= -\frac{1}{x}(x-1)^2 - \sum_{i=1}^4 \left(1- \mathcal R_i  \right) y_i - \sum_{i=1}^2  s_iz_i
\end{align*}

\textbf{Case ii.}:
\begin{align*}
\dot W= -\frac{1}{\mathcal R_1 x}(\mathcal R_1 x-1)^2 - \sum_{i=2}^4 \frac{y_i}{\mathcal R_0} \left(\mathcal R_1- \mathcal R_i  \right)  - \sum_{i=1}^2  s_iz_i
\end{align*}

\textbf{Case iii.}:
\begin{align*}
\dot W= -\frac{1}{\mathcal Q_1 x}(\mathcal Q_1 x-1)^2 - \frac{y_3}{\mathcal Q_1} \left(\mathcal R_1- \mathcal R_3  \right) - \sum_{i=2,4} \frac{y_i}{\mathcal Q_1} \left(\mathcal Q_1- \mathcal R_i  \right)  - z_2(s_2-s_1)
\end{align*}

\textbf{Case iv.}:
\begin{align*}
\dot W= -\frac{1}{\mathcal R_2 x}(\mathcal R_2 x-1)^2 - \frac{y_3}{\mathcal R_2} \left(\mathcal R_1- \mathcal R_3  \right) - \frac{y_4}{\mathcal R_2} \left(\mathcal R_2- \mathcal R_4  \right)  - \frac{z_2}{\mathcal R_2}\left(\mathcal Q_2-\mathcal R_2\right)
\end{align*}

In case iii, on the attracting invariant set $\mathcal L$ (same $\mathcal L$ as in proof of Theorem \ref{genThm}), we have $x=1/\mathcal Q_1$, $y_3=y_4=0$ and $z_2=0$.  If $\mathcal R_1<\mathcal Q_1$, then $\dot W=0 \Rightarrow y_2=0$, otherwise $\mathcal R_1=\mathcal Q_1\Rightarrow \dot y_2=0$.  Either way the equation $\dot x=0$ implies that $y_1=y_1^*$ and $y_2=0$.   Thus $\dot y_1=0$, which implies $z_1=z_1^*$.  So the invariant omega limit set, which is nonempty by compactness of orbits, is
$\mathcal L = \{\bar{\mathcal E}_1\}$ and we get convergence.

In case iv, i.e. $\mathcal Q_1<\mathcal R_2 \leq \mathcal Q_2$, then we prove $\widetilde{\mathcal E}_1$ is GAS.  Here, $\dot W=0$ iff $x=x^*=\frac{1}{\mathcal R_{1} }$, $y_3=y_4=0$, and $z_2=0$.  (Note if $\mathcal Q_2>\mathcal R_2$, $z_2=0$ is immediate.  If not, we can still reason that the asymptotic average of $z_2$ must converge to $z_2^*=0$, hence $z_2=0$.)  In addition by Theorem \ref{genThm}, we can obtain  $y_2=y_2^*=1-\frac{\mathcal Q_1}{\mathcal R_1}$.  The last relation combined with $x=x^*$ implies that $y_1=y_1^*=s_1$.  Then $\dot{y}_1=0$ implies $z_1=z_1^*$.   Thus $\mathcal L$ consists solely of the equilibrium $\widetilde{\mathcal E}_1$.

Next in Theorem \ref{mainThm}:

\textbf{Case 1}:
$(x^*,y^*,z^*)=\bar{\mathcal E}_2$:
\begin{align*}
\dot W &= \frac{-1}{\mathcal Q_2 x}\left(\mathcal Q_2 x -1\right)^2 - \frac{y_3}{\mathcal Q_2} \left( \mathcal Q_2 - \mathcal R \right)- \frac{y_4}{\mathcal Q_2} \left( \mathcal Q_2 - \mathcal R_4 \right).
\end{align*}
Notice that $\dot W\leq 0$ when $\mathcal R\leq \mathcal Q_2$ and $\mathcal R_4\leq \mathcal Q_2$.  Also, the equilibrium $\bar{\mathcal E}_2$ is non-negative when $\mathcal R_2> \mathcal Q_2$.   If $\mathcal R < \mathcal Q_2$ and $\mathcal R_4 < \mathcal Q_2$, applying Theorem \ref{genThm}, all inequalities (\ref{inequ}) are strict and only two strains, $y_1$ and $y_2$ have non-empty epitope sets,  and therefore  $\bar{\mathcal E}_2$ is globally asymptotically stable.
If $\mathcal R_4 =\mathcal Q_2$, then the differential equations in (\ref{invSys}) hold along with $\dot y_4=0$ and $\sum_{i=1,2,4} \mathcal R_i  y_i=\mathcal Q_2-1$.  Taking asymptotic averages as in the proof of Theorem 3.1 in \cite{browneNest}, we obtain that indeed $y_4=0$, and we similarly obtain that  $\bar{\mathcal E}_2$ is globally asymptotically stable.

\textbf{Case 2}:
$(x^*,y^*,z^*)=\widetilde{\mathcal E}_2$:
\begin{align*}
\dot W &= \frac{-1}{\mathcal \mathcal R_4 x}\left(\mathcal R_4 x -1\right)^2 - \frac{y_3}{\mathcal R_4} \left( \mathcal R_4 - \mathcal R \right),
\end{align*}
Notice that $\dot W\leq 0$ when $\mathcal R< \mathcal R_4$.  Also, the equilibrium $\widetilde{\mathcal E}_2$ is non-negative when $\mathcal R_4> \mathcal Q_2$.  Applying Theorem \ref{genThm}, we obtain that $\widetilde{\mathcal E}_2$ is globally asymptotically stable.

\textbf{Case 3}:
$(x^*,y^*,z^*)=\widehat{\mathcal E}_2$:
\begin{align*}
\dot W &= \frac{-1}{\mathcal R x}\left(\mathcal R x -1\right)^2 - \frac{y_4}{\mathcal R} \left( \mathcal R - \mathcal R_4 \right),
\end{align*}
Notice that $\dot W\leq 0$ when $\mathcal R_4 \leq \mathcal R$.  Also, the equilibrium $\widehat{\mathcal E}_2$ is non-negative when $ \mathcal Q_2 < \mathcal R < \mathcal P_2$.  The result is a direct consequence of Theorem \ref{genThm}.

\textbf{Case 4}:
$(x^*,y^*,z^*)=\mathcal E^{\dagger}_2$:
\begin{align*}
\dot W &= \frac{-1}{\mathcal P_2 x}\left(\mathcal P_2 x -1\right)^2 - \frac{y_1}{\mathcal P_2} \left( \mathcal R - \mathcal P_2 \right)- \frac{y_4}{\mathcal P_2} \left( \mathcal P_2 - \mathcal R_4 \right),
\end{align*}
Notice that $\dot W\leq 0$ when $\mathcal R \geq \mathcal P_2$ and $\mathcal R_4 \leq \mathcal P_2$.  The equilibrium $\mathcal E^{\dagger}_2$ is non-negative when $\min(\mathcal R_2, \mathcal R_3)> \mathcal P_2$.  Global stability follows from Theorem \ref{genThm}.

\textbf{Case 5}:
$(x^*,y^*,z^*)=\mathcal E^{\ddagger}_2$:
\begin{align*}
\dot W &= \frac{-1}{\mathcal \mathcal R_4 x}\left(\mathcal R_4 x -1\right)^2 - \frac{y_1}{\mathcal R_4} \left( \mathcal R - \mathcal R_4 \right),
\end{align*}
Notice that $\dot W\leq 0$ when $\mathcal R_4 \leq \mathcal R$.  Also, the equilibrium $\mathcal E^{\dagger}_2$ is non-negative when $\mathcal R_4> \mathcal P_2$.  Global stability follows from Theorem \ref{genThm}.

\textbf{Case 6}:
$(x^*,y^*,z^*)=\mathcal E(y_4^*)$ ($\mathcal R_4=\mathcal R$):
\begin{align*}
\dot W &= \frac{-1}{\mathcal R x}\left(\mathcal R x -1\right)^2 ,
\end{align*}
Notice that $\dot W\leq 0$ and the equilibrium $\mathcal E(y_4^*)$ is non-negative when $ \mathcal Q_2 \leq \mathcal R < \mathcal P_2$.  Thus $\mathcal E(y_4^*)$ is saturated, though not strictly.  Local stability follows from Theorem \ref{genThm}.
\end{proof}

In Figure \ref{fig3}, we illustrate the bifurcation at $\mathcal R_4=\mathcal R$.  

\renewcommand{\thefigure}{A\arabic{figure}}

\setcounter{figure}{0}

 \begin{figure}[t!]
\subfigure[][]{\label{fig3a} \includegraphics[height=3cm,width=4.5cm]{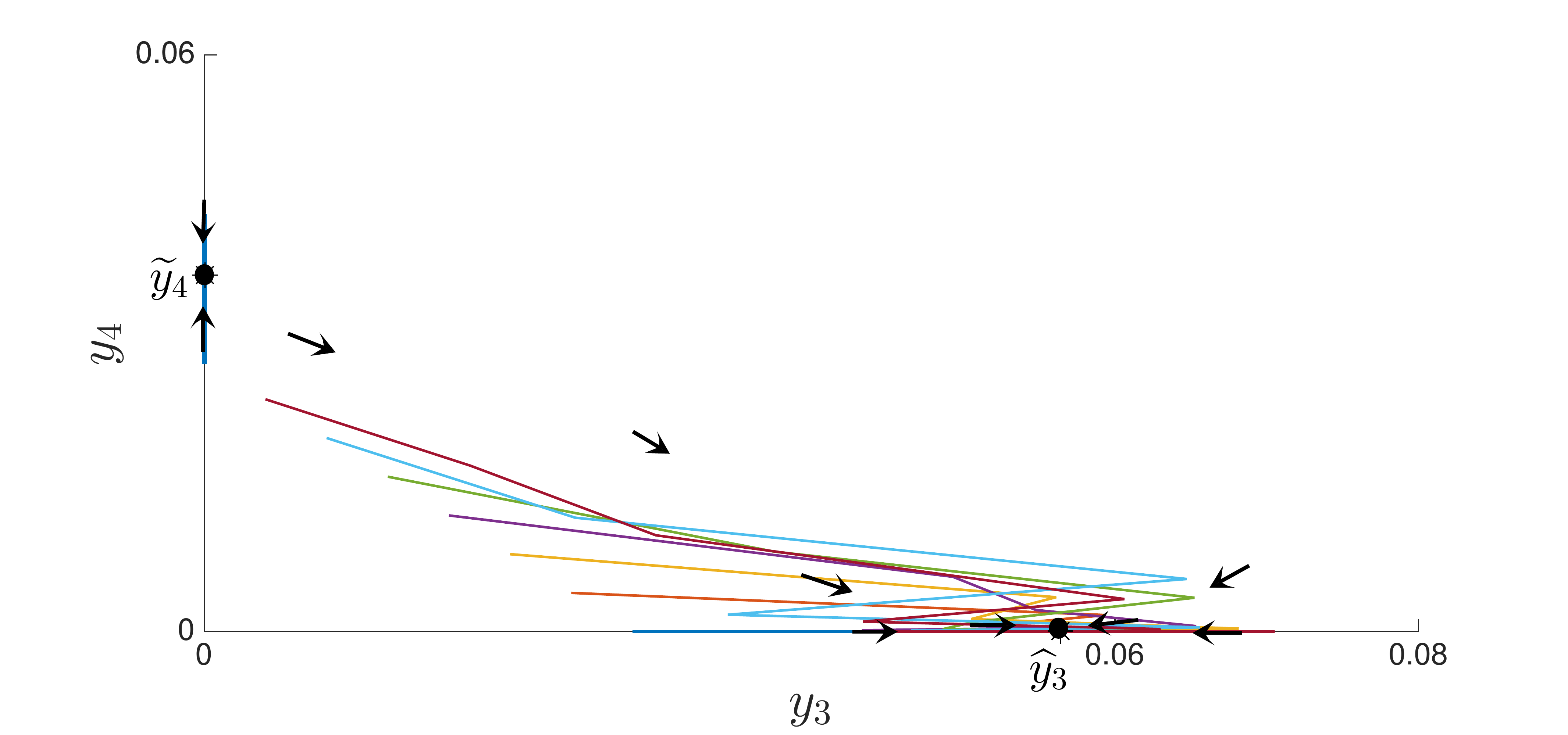}}
\subfigure[][]{\label{fig3b} \includegraphics[height=3cm,width=4.7cm]{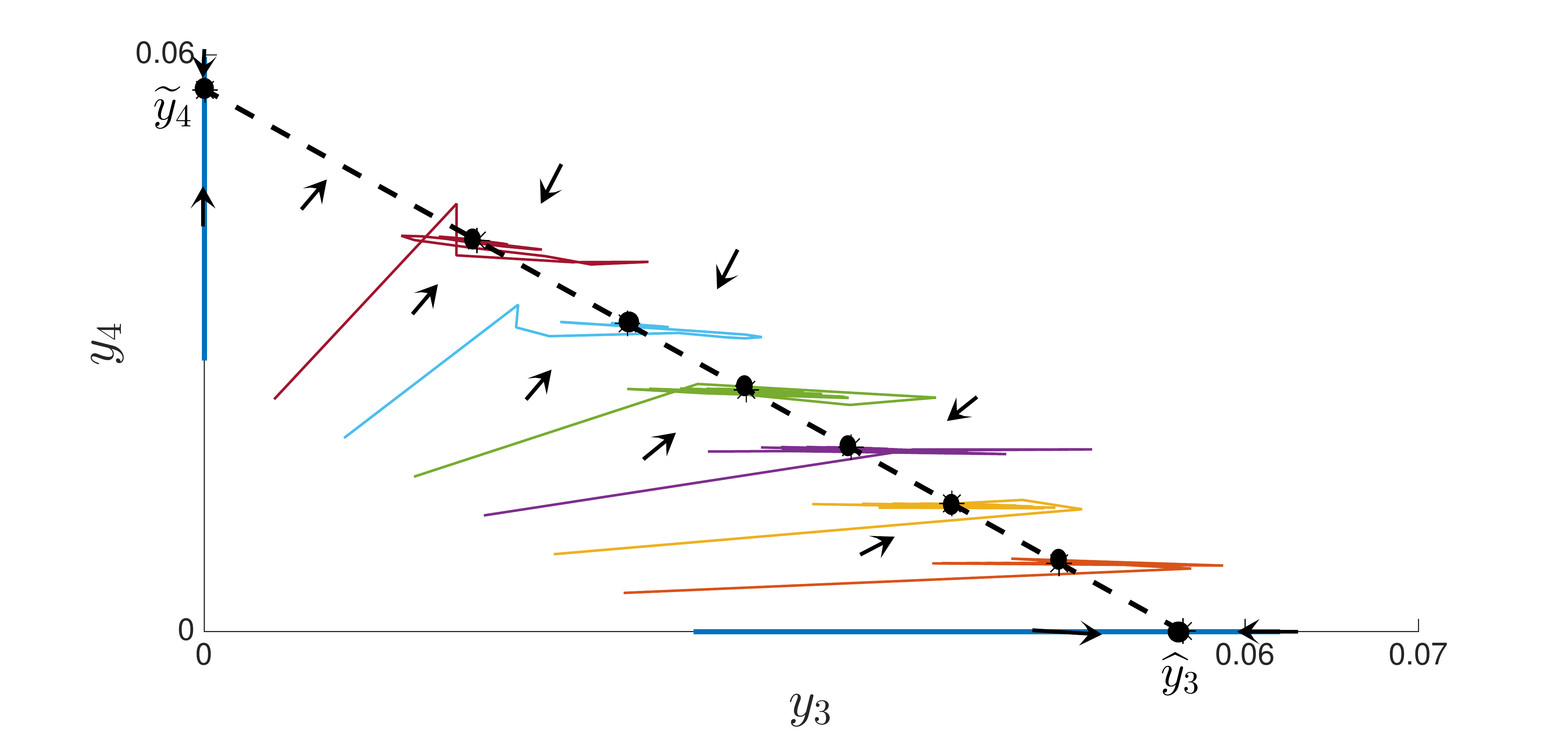}}
\subfigure[][]{\label{fig3c} \includegraphics[height=3cm,width=4.5cm]{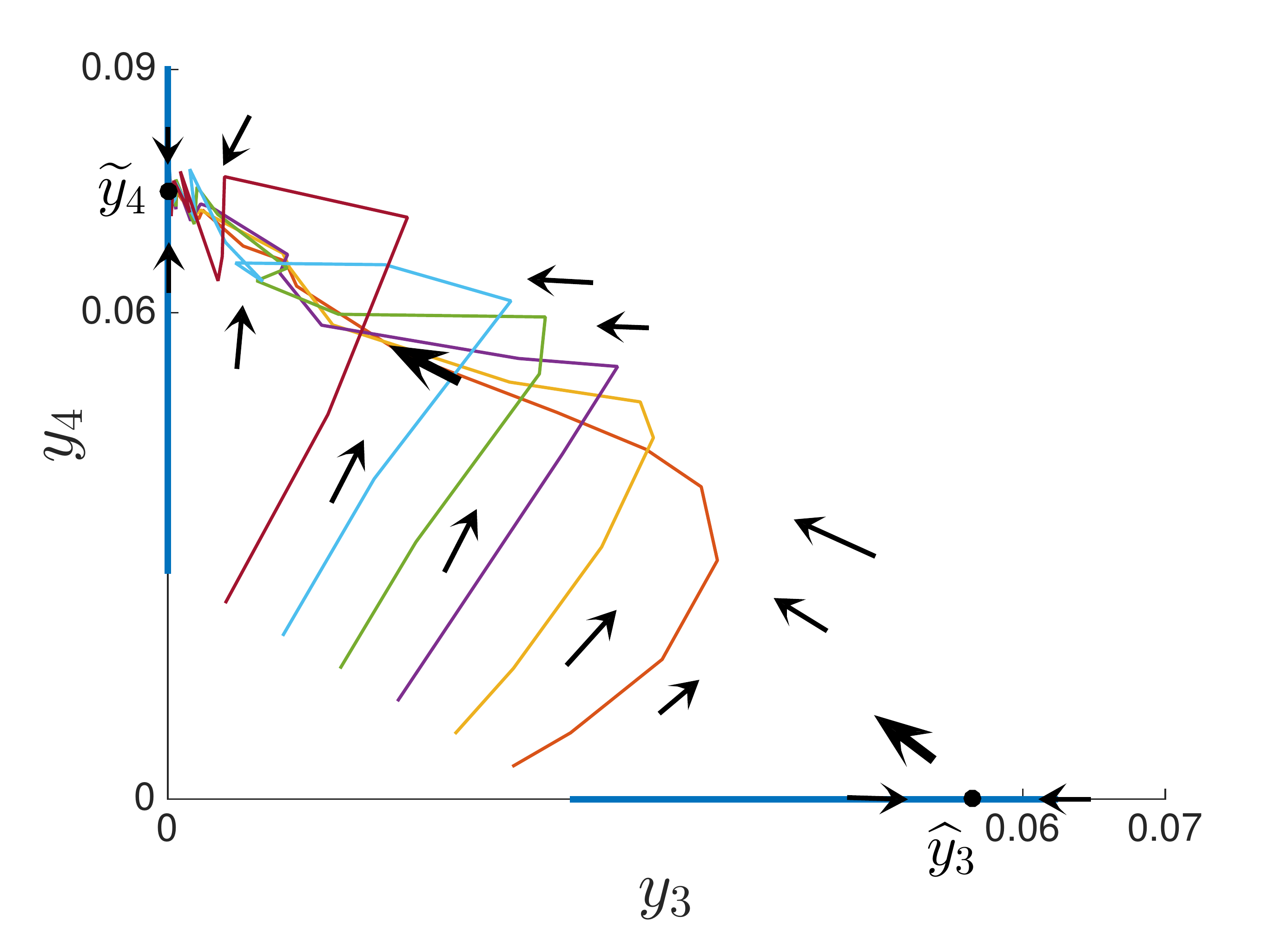}}
\caption{\footnotesize{Bifurcation at $\mathcal R_4=\mathcal R$: line of equilibria (a) $\mathcal R_4<\mathcal R \Rightarrow$ $y_3\rightarrow \widehat y_3, y_4\rightarrow 0$ (b) Line of equilibria at $\mathcal R_4=\mathcal R$ where all variants persist. (c) $\mathcal R_4>\mathcal R \Rightarrow$ $y_3\rightarrow 0, y_4\rightarrow \widetilde y_4$. }}
  \label{fig3}
  \end{figure}

\subsection{Two-epitope model: dynamics when $s_1=s_2$} \label{A3}
We analyze the feasible stable equilibria with immune response for the case of equal immunodominance of $z_1$ and $z_2$ ($s_1=s_2=s$) in model (\ref{ode4}).  First, consider equilibria with one immune response present.  Since $s_1=s_2$, without loss of generality, we can take $z_1^*>0$.  From the equilibrium equations, $y_1^*,y_3^*>0 \Leftrightarrow \mathcal R_1=\mathcal R_3$, which is not possible since $\mathcal R_1>\mathcal R_3$.  Also, an equilibrium with $y_3^*>0$ (and $z_2^*=0$) will not be saturated since $\mathcal R_1>\mathcal R_3$, therefore take $y_3^*=0$.  Similarly we can take $y_4^*=0$. Thus, consider equilibria $\widetilde{\mathcal E}_1$ and $\overline{\mathcal E}_1$.  The equilibrium $\widetilde{\mathcal E}_1$ can not be stable since $s_1=s_2\Rightarrow \mathcal Q_1=\mathcal Q_2$, which does not permit conditions for case iv of Theorem \ref{oneThm} to be satisfied.  The equilibrium $\overline{\mathcal E}_1$ will not be ``strictly saturated'' under conditions for case iii of Theorem \ref{oneThm} because there are a continuum of equilibria of the form $\overline{\mathcal E}_1(z_2^*)$ where $\bar z_1=\frac{\mathcal R_1}{\mathcal Q_1}-1-z_2^*, \bar z_2=z_2^*$.

Now consider the possibility of equilibria with $z_1^*,z_2^*>0$.  Observe that at least two of the viral strain components $y_1^*,y_2^*,y_3^*$ are positive and $y_2^*=y_3^*$ from the $\dot z_1, \dot z_2$ equations.  Therefore $y_2^*,y_3^*>0$ and for the case $s_1=s_2$, non-trivial strictly saturated equilibria are of the form $\widehat{\mathcal E}_2$,   $\mathcal E_2^{\dagger}$ or $\mathcal E_2^{\ddagger}$ with the corresponding parameter regimes as defined in Table \ref{table}.

\subsection{Instability of ``one-mutation'' and ``strain-specific'' equilibria for (\ref{odeU})} \label{A4}
We assume equal viral fitness costs for mutation from each of $n$ epitopes, yielding system (\ref{odeU}), as described in Section \ref{EqFitSec}.  Consider the set $\mathcal S_1$ of $n$ viral strains which have exactly one mutation, i.e. $y_i\in \mathcal S_1 \Rightarrow d(y_i,y_1)=1$.  For clarity here, we label these strains as $\mathcal S_1=\left\{y_i^1 \ | \  i=1,\dots,n\right\}$ where $y_i^1$ has escaped $z_i$ but is susceptible all other immune responses.  Note that $y_1^1$ is strain $y_2$ with our ``nested'' indexing introduced in Section \ref{NestedSec}.  The subsystem only containing these viral strains looks as follows:
\begin{align}
\dot x &= 1-x- x\sum_{i=1}^n f\mathcal R_1 y_{i_1}, \quad \dot y_i^1 = \gamma_i^1 y_i^1\left(f\mathcal R_1 x -1 -\sum_{j\neq i}  z_j \right), \quad  \dot z_i = \frac{\sigma_i}{s_i} z_i\left(\sum_{j\neq i} y_j^1 -s_i \right),  \quad i=1,\dots,n. \label{odeS1}
\end{align}
By Proposition \ref{prop33} and \eqref{genEqcon}, a positive equilibrium $\mathcal E^*_1=(x^*,y^*,z^*)$ of system (\ref{odeS1}) satisfies
\begin{align*}
 x^*=\frac{1}{1+\vec s^T A^{-1} f\mathcal R_1\vec 1}, \quad A y^* &= \vec s,  \quad Az^*=(f\mathcal R_1x^*-1)\vec 1,  \\
\text{where} \quad & A=\vec 1 \left(\vec 1\right)^T - I_n, \quad A^{-1}=\frac{1}{n-1}\vec 1 \left(\vec 1\right)^T - I_n, \quad y=\left(y_1^1,y_2^1,\dots,y_n^1\right)^T,
\end{align*}
where $I_n$ is the $n\times n$ identity matrix.  Here we find that:
\begin{align*}
y_i^{1*} &= \frac{1}{n-1}\left(-(n-2)s_i+\sum_{j\neq i}s_j\right), \quad x^*=\frac{n-1}{n-1+f\mathcal R_1 \sum_i s_i}, \quad z^*_i=\frac{1}{n-1}(f\mathcal R_1x^*-1)
\end{align*}
Assuming our hierarchy $s_i\leq s_{i+1}$, then $y_i^{1*}>0$ if $s_1>\sum_{i>1}(s_n-s_i)$ and $z_i^*>0$ if $f\mathcal R_1\left(n-1-\sum_i s_i\right)>n-1$.  If these conditions are satisfied, then the equilibrium $\mathcal E^*_1$ is saturated in the subsystem (\ref{odeS1}) where $y_i\in\mathcal S_1$.  However, if we consider a larger network of viral strains, then equilibrium $\mathcal E^*_1$ is always unstable in this case with equal fitness costs for mutation.  Indeed, consider the wild strain $y_1\sim 0\cdots 0$ and strain $y_{3}\sim 110\cdots 0$, which correspond to backward and forward mutations from the strain $y_1^1\in\mathcal S_1$.  We calculate their invasion rates at equilibrium $\mathcal E^*_1$.  Note that $y_1$ has reproduction number $\mathcal R_1$ and $y_{3}$ has reproduction number $\mathcal R_{3}=f^2\mathcal R_{1}$.  In order for $\mathcal E^*_1$ to be saturated (stable), we find that:
\begin{align*}
&\frac{\dot y_1}{\gamma_1 y_1}= \mathcal R_1x^* -1 - \sum_{i=1}^n z_i^* \leq 0, \qquad
\frac{\dot y_{3}}{\gamma_3 y_{3}}= f^2 \mathcal R_1x^* -1\sum_{i=3}^n z_i^* \leq 0 \\
 &\Leftrightarrow\mathcal R_1x^* -1 - \frac{n}{n-1}\left(f\mathcal R_1x^*-1\right) \leq 0, \qquad
f^2 \mathcal R_1x^* - 1 - \frac{n-2}{n-1}\left(f\mathcal R_1x^*-1\right) \leq 0 \\
&\Leftrightarrow \mathcal R_1x^*\left(fn-(n-1)\right) \geq 1, \qquad
f\mathcal R_1x^*\left(f(n-1)-(n-2)\right)  \leq 1 \\
&\Leftrightarrow f>\frac{n-1}{n}, \quad  f((n-1)f-(n-2))\leq nf-(n-1) \\
&\Leftrightarrow  0  \geq (n-1)(f-1)^2
\end{align*}
 However, $(n-1)(f-1)^2>0$  for all $f\neq1, n>1$, giving a contradiction. Thus $\mathcal E_1$ can only be stable when restricted to the subsystem (\ref{odeS1}) consisting of strains in $\mathcal S_1$, and becomes unstable in the larger network in this case.

Similarly, we can show that the ``strain-specific'' equilibria, $\mathcal E^{\dagger}_n$ and  $\mathcal E^{\ddagger}_{n+1}$ (introduced in Section (\ref{ssSec})), are always unstable in this case of uniform fitness costs.  Indeed, for clarity here, denote $y^o_i$ as the viral strain with epitope set $\Lambda^o_i=\left\{i\right\}$, and $y_{n+1}$ as the strain which has completely escaped all $z_i$, i.e. $\Lambda_{n+1}=\emptyset$.  Then $\mathcal E^{\dagger}_n$ (and  $\mathcal E^{\ddagger}_{n+1}$) consist of equilibria where $y^{o*}_i>0$ (and $y_{n+1}^*>0$).  First, consider the case that $\mathcal R_{n+1}=f^n\mathcal R_1\leq \mathcal P_n^o=1+\sum_i s_i\mathcal R^o_i=1+f^{n-1}\mathcal R_1\sum_i s_i$ and $f^{n-1}\mathcal R_1>\mathcal P_n^o$, so that $\mathcal E^{\dagger}_n$ is positive, but  $\mathcal E^{\ddagger}_n$ is not positive.  Then consider the invasion rate of viral strain $y_{n-1}\sim 1\cdots 100$ (with $\Lambda_{n-1}= \left\{n-1,n\right\}$ and reproduction number $\mathcal R_{n-1}=f^{n-2}\mathcal R_1$).  Utilizing Proposition \ref{ssF}, if $\mathcal E^{\dagger}_n$ is stable, then
\begin{align*}
\mathcal P_n^o + f^{n-2}\mathcal R_1&\leq 2f^{n-1}\mathcal R_1 \\
\Leftrightarrow  \mathcal R_1 f^{n-2}\left( f^2 -2f +1\right) &\leq 0, \quad \text{since} \ \ f^n\mathcal R_1\leq \mathcal P_n^o,
\end{align*}
which is clearly a contradiction since $\left( f^2 -2f +1\right)=(f-1)^2>0$ for $f<1$.  A similar argument applies to $\mathcal E^{\ddagger}_{n+1}$ in the case $f^n\mathcal R_1>\mathcal P_n^o$.  Thus the ``strain-specific'' equilibria are always unstable for system (\ref{odeU}).

\bibliography{hivctlNetwork_References.bib}

\begin{thebibliography}{10}
\providecommand{\url}[1]{{#1}}
\providecommand{\urlprefix}{URL }
\expandafter\ifx\csname urlstyle\endcsname\relax
  \providecommand{\doi}[1]{DOI~\discretionary{}{}{}#1}\else
  \providecommand{\doi}{DOI~\discretionary{}{}{}\begingroup
  \urlstyle{rm}\Url}\fi

\bibitem{Althaus}
Althaus, C.L., Boer, R.D.: Dynamics of immune escape during
  \uppercase{hiv}/\uppercase{siv} infection.
\newblock PLoS Computational Biology \textbf{4}, e1000,103 (2008)

\bibitem{asquith2006inefficient}
Asquith, B., Edwards, C.T., Lipsitch, M., McLean, A.R.: Inefficient cytotoxic t
  lymphocyte--mediated killing of \uppercase{hiv}-1--infected cells in vivo.
\newblock PLoS Biol \textbf{4}(4), e90 (2006)

\bibitem{Barton}
Barton, J.P., Goonetilleke, N., Butler, T.C., Walker, B.D., McMichael, A.J.,
  Chakraborty, A.K.: Relative rate and location of intra-host \uppercase{hiv}
  evolution to evade cellular immunity are predictable.
\newblock Nature communications \textbf{7} (2016)

\bibitem{bobko2015singularly}
Bobko, N., Zubelli, J.P.: A singularly perturbed \uppercase{hiv} model with
  treatment and antigenic variation.
\newblock Mathematical biosciences and engineering: MBE \textbf{12}(1), 1--21
  (2015)

\bibitem{browne2016}
Browne, C.: Immune response in virus model structured by cell infection-age.
\newblock Mathematical biosciences and engineering: MBE \textbf{13}(5) (2016)

\bibitem{browne2016global}
Browne, C.: Global properties of nested network model with application to
  multi-epitope \uppercase{hiv}/\uppercase{ctl} dynamics.
\newblock Journal of Mathematical Biology pp. 1--22 (2017)

\bibitem{deBoer}
De~Boer, R.J.: Which of our modeling predictions are robust.
\newblock PLoS Comput Biol \textbf{8}(7), e1002,593 (2012)

\bibitem{vanDeutekom}
Deutekom, H.V., Wijnker, G., Boer, R.D.: The rate of immune escape vanishes
  when multiple immune responses control an \uppercase{hiv} infection.
\newblock Journal of immunology \textbf{191}, 3277--3286 (2013)

\bibitem{Vitaly3}
Ganusov, V.V., De~Boer, R.J.: Estimating costs and benefits of \uppercase{ctl}
  escape mutations in \uppercase{siv}/\uppercase{hiv} infection.
\newblock PLoS computational biology \textbf{2}(3), e24 (2006)

\bibitem{Vitaly1}
Ganusov, V.V., Goonetilleke, N., Liu, M.K., Ferrari, G., Shaw, G.M., Borrow,
  A.J.M.P., Korber, B.T., Perelson, A.S.: Fitness costs and diversity of the
  cytotoxic t lymphocyte (\uppercase{ctl}) response determine the rate of
  \uppercase{ctl} escape during acute and chronic phases of \uppercase{hiv}
  infection.
\newblock Journal of virology \textbf{85}(20), 10,518--10,528 (2011)

\bibitem{Vitaly2}
Ganusov, V.V., Neher, R.A., Perelson, A.S.: Mathematical modeling of escape of
  \uppercase{hiv} from cytotoxic t lymphocyte responses.
\newblock Journal of Statistical Mechanics: Theory and Experiment
  \textbf{2013.01}, P01,010 (2013)

\bibitem{Goh}
Goh, B.: Sector stability of a complex ecosystem model.
\newblock Mathematical Biosciences \textbf{40}(1-2), 157--166 (1978)

\bibitem{gurney2017network}
Gurney, J., Aldakak, L., Betts, A., Gougat-Barbera, C., Poisot, T., Kaltz, O.,
  Hochberg, M.E.: Network structure and local adaptation in co-evolving
  bacteria--phage interactions.
\newblock Molecular ecology \textbf{26}(7), 1764--1777 (2017)

\bibitem{hofbauer1990index}
Hofbauer, J.: An index theorem for dissipative semiflows.
\newblock The Rocky Mountain Journal of Mathematics \textbf{20}(4), 1017--1031
  (1990)

\bibitem{hofbauer1998evolutionary}
Hofbauer, J., Sigmund, K.: Evolutionary games and population dynamics.
\newblock Cambridge university press (1998)

\bibitem{iwasa2004some}
Iwasa, Y., Michor, F., Nowak, M.: Some basic properties of immune selection.
\newblock Journal of theoretical biology \textbf{229}(2), 179--188 (2004)

\bibitem{jover2013mechanisms}
Jover, L.F., Cortez, M.H., Weitz, J.S.: Mechanisms of multi-strain coexistence
  in host--phage systems with nested infection networks.
\newblock Journal of theoretical biology \textbf{332}, 65--77 (2013)

\bibitem{kessinger2015inferring}
Kessinger, T.A., Perelson, A.S., Neher, R.A.: Inferring \uppercase{hiv} escape
  rates from multi-locus genotype data.
\newblock Immune system modeling and analysis p. 348 (2015)

\bibitem{Kloverpris}
Kl{\o}verpris, H.N., Payne, R.P., Sacha, J.B., Rasaiyaah, J.T., Chen, F.,
  Takiguchi, M., Yang, O.O., Towers, G.J., Goulder, P., Prado, J.G.: Early
  antigen presentation of protective \uppercase{hiv}-1 kf11gag and kk10gag
  epitopes from incoming viral particles facilitates rapid recognition of
  infected cells by specific cd8+ t cells.
\newblock Journal of virology \textbf{87}(5), 2628--2638 (2013)

\bibitem{korytowski2015nested}
Korytowski, D.A., Smith, H.L.: How nested and monogamous infection networks in
  host-phage communities come to be.
\newblock Theoretical ecology \textbf{8}(1), 111--120 (2015)

\bibitem{Leviyang}
Leviyang, S.: The coalescence of intrahost \uppercase{hiv} lineages under
  symmetric \uppercase{ctl} attack.
\newblock Bulletin of mathematical biology \textbf{74}(8), 1818--1856 (2012)

\bibitem{liu2013vertical}
Liu, M.K., Hawkins, N., Ritchie, A.J., Ganusov, V.V., Whale, V., Brackenridge,
  S., Li, H., Pavlicek, J.W., Cai, F., Rose-Abrahams, M., et~al.: Vertical t
  cell immunodominance and epitope entropy determine \uppercase{hiv}-1 escape.
\newblock The Journal of clinical investigation \textbf{123}(1), 380--393
  (2013)

\bibitem{luo2015competitive}
Luo, S., Perelson, A.S.: Competitive exclusion by autologous antibodies can
  prevent broad \uppercase{hiv}-1 antibodies from arising.
\newblock Proceedings of the National Academy of Sciences \textbf{112}(37),
  11,654--11,659 (2015)

\bibitem{korytowski2017persistence}
\MakeLowercase{KORYTOWSKI, DAN A and SMITH, HAL L}: Persistence in
  phage-bacteria communities with nested and one-to-one infection networks.
\newblock Discrete \& Continuous Dynamical Systems-Series B \textbf{22}(3)
  (2017)

\bibitem{nowak1996population}
Nowak, M.A., Bangham, C.R.: Population dynamics of immune responses to
  persistent viruses.
\newblock Science \textbf{272}(5258), 74--79 (1996)

\bibitem{pandit2014reliable}
Pandit, A., de~Boer, R.J.: Reliable reconstruction of \uppercase{hiv}-1 whole
  genome haplotypes reveals clonal interference and genetic hitchhiking among
  immune escape variants.
\newblock Retrovirology \textbf{11}(1), 1 (2014)

\bibitem{Perelson2}
Perelson, A.S., Nelson, P.W.: Mathematical analysis of \uppercase{hiv}-1
  dynamics in vivo.
\newblock SIAM review \textbf{41}(1), 3--44 (1999)

\bibitem{pruss2008global}
Pr{\"u}ss, J., Zacher, R., Schnaubelt, R.: Global asymptotic stability of
  equilibria in models for virus dynamics.
\newblock Mathematical Modelling of Natural Phenomena \textbf{3}(7), 126--142
  (2008)

\bibitem{smith2011dynamical}
Smith, H.L., Thieme, H.R., Thieme, H.R.: Dynamical systems and population
  persistence, vol. 118.
\newblock American Mathematical Society Providence, RI (2011)

\bibitem{souza2011global}
Souza, M.O., Zubelli, J.P.: Global stability for a class of virus models with
  \uppercase{cytotoxic T lymphocyte} immune response and antigenic variation.
\newblock Bulletin of mathematical biology \textbf{73}(3), 609--625 (2011)

\bibitem{takeuchi1996global}
Takeuchi, Y.: Global dynamical properties of Lotka-Volterra systems.
\newblock World Scientific (1996)

\bibitem{thieme1993persistence}
Thieme, H.R.: Persistence under relaxed point-dissipativity (with application
  to an endemic model).
\newblock SIAM Journal on Mathematical Analysis \textbf{24}(2), 407--435 (1993)

\bibitem{weitz2013phage}
Weitz, J.S., Poisot, T., Meyer, J.R., Flores, C.O., Valverde, S., Sullivan,
  M.B., Hochberg, M.E.: Phage--bacteria infection networks.
\newblock Trends in microbiology \textbf{21}(2), 82--91 (2013)

\bibitem{wolkowicz1989successful}
Wolkowicz, G.S.: Successful invasion of a food web in a chemostat.
\newblock Mathematical Biosciences \textbf{93}(2), 249--268 (1989)

\end{thebibliography}
\bibliographystyle{spmpsci}
\end{document}